\documentclass[11pt]{article}
\pdfoutput=1
\usepackage[T1]{fontenc}
\usepackage{lmodern}
\usepackage{slantsc}
\usepackage[protrusion=true,expansion=true]{microtype}
\usepackage{breakcites}
\usepackage{amsmath,amssymb,amsfonts,amsthm}
\usepackage{subcaption}
\usepackage{graphicx}
\usepackage{fullpage}
\usepackage{setspace}
\usepackage[backref=page]{hyperref}
\usepackage[nameinlink]{cleveref}
\usepackage{color}
\usepackage{wrapfig}
\usepackage{tikz}
\usetikzlibrary{decorations.pathreplacing}
\usepackage{algorithm}
\usepackage[noend]{algpseudocode}
\usepackage[framemethod=tikz]{mdframed}
\usepackage{xspace}
\usepackage{pgfplots}
\usepackage{framed}
\usepackage{thmtools}
\usepackage{thm-restate}
\usepackage{tabu}
\usepackage{fancyhdr}
\usepackage{bbm}
\usepackage{booktabs}
\pgfplotsset{compat=1.5}

\newtheorem{theorem}{Theorem}[section]

\newtheorem{lemma}[theorem]{Lemma}

\newtheorem{definition}[theorem]{Definition}

\newtheorem{fact}[theorem]{Fact}

\newenvironment{proofof}[1]{\begin{trivlist} \item {\bf Proof
#1:~~}}
  {\qed\end{trivlist}}

\newcommand{\namedref}[2]{\hyperref[#2]{#1~\ref*{#2}}}

\renewcommand{\Re}{\operatorname{Re}}
\renewcommand{\Im}{\operatorname{Im}}

\def \FAIL    {\mdef{\mathsf{FAIL}}}

\def \CountSketch    {\mdef{\textsc{CountSketch}}}

\newcommand{\PPr}[1]{\ensuremath{\mathbf{Pr}\left[#1\right]}}
\newcommand{\PPPr}[2]{\ensuremath{\underset{#1}{\mathbf{Pr}}\left[#2\right]}}
\newcommand{\Ex}[1]{\ensuremath{\mathbb{E}\left[#1\right]}}
\newcommand{\EEx}[2]{\ensuremath{\underset{#1}{\mathbb{E}}\left[#2\right]}}
\renewcommand{\O}[1]{\ensuremath{\mathcal{O}\left(#1\right)}}
\newcommand{\tO}[1]{\ensuremath{\tilde{\mathcal{O}}\left(#1\right)}}
\newcommand{\eps}{\varepsilon}

\def \calD    {\mdef{\mathcal{D}}}
\def \calE    {\mdef{\mathcal{E}}}

\def \calH    {\mdef{\mathcal{H}}}

\def \calN    {\mdef{\mathcal{N}}}

\def \calP    {\mdef{\mathcal{P}}}

\def \calS    {\mdef{\mathcal{S}}}

\def \calZ    {\mdef{\mathcal{Z}}}

\newcommand{\mdef}[1]{{\ensuremath{#1}}\xspace}  

\DeclareMathOperator*{\polylog}{polylog}
\DeclareMathOperator*{\poly}{poly}

\DeclareMathOperator*{\Exp}{Exp}
\DeclareMathOperator{\Si}{Si}
\DeclareMathOperator{\Cin}{Cin}
\DeclareMathOperator{\tail}{tail}
\DeclareMathOperator{\Unif}{Unif}
\DeclareMathOperator*{\Ein}{Ein}

\newcommand{\ignore}[1]{}

\newif\ifnotes\notestrue 
\ifnotes
\newcommand{\samson}[1]{\textcolor{blue}{{\bf (Samson:} {#1}{\bf ) }} \marginpar{\tiny\bf
             \begin{minipage}[t]{0.5in}
               \raggedright S:
            \end{minipage}}}
\newcommand{\david}[1]{\textcolor{purple}{{\bf (David:} {#1}{\bf ) }} \marginpar{\tiny\bf
             \begin{minipage}[t]{0.5in}
               \raggedright D:
            \end{minipage}}} 
\else
\newcommand{\samson}[1]{}
\newcommand{\david}[1]{}
\fi

\makeatletter
\renewcommand*{\@fnsymbol}[1]{\textcolor{mahogany}{\ensuremath{\ifcase#1\or *\or \dagger\or \ddagger\or
 \mathsection\or \triangledown\or \mathparagraph\or \|\or **\or \dagger\dagger
   \or \ddagger\ddagger \else\@ctrerr\fi}}}
\makeatother

\providecommand{\email}[1]{\href{mailto:#1}{\nolinkurl{#1}\xspace}}

\definecolor{mahogany}{rgb}{0.75, 0.25, 0.0}
\definecolor{darkblue}{rgb}{0.0, 0.0, 0.55}
\definecolor{darkpastelgreen}{rgb}{0.01, 0.75, 0.24}
\definecolor{darkgreen}{rgb}{0.0, 0.2, 0.13}
\definecolor{darkgoldenrod}{rgb}{0.72, 0.53, 0.04}
\definecolor{darkred}{rgb}{0.55, 0.0, 0.0}
\definecolor{forestgreenweb}{rgb}{0.13, 0.55, 0.13}
\definecolor{greencss}{rgb}{0.0, 0.5, 0.0}
\definecolor{bleudefrance}{rgb}{0.19, 0.55, 0.91}
\definecolor{darkpastelpurple}{rgb}{0.59, 0.44, 0.84}
\definecolor{darkcerulean}{rgb}{0.03, 0.27, 0.49}
\definecolor{royalpurple}{rgb}{0.47, 0.32, 0.66}

\hypersetup{
     colorlinks   = true,
     citecolor    = mahogany,
	 linkcolor	  = forestgreenweb
}

\begin{document}

\title{Fast Perfect $L_2$ Sampling}
\author{David P. Woodruff\thanks{Carnegie Mellon University and Google Research. \email{dwoodruf@andrew.cmu.edu}. David P. Woodruff is supported in part by Office of Naval Research award number N000142112647 and a Simons Investigator Award. }
\and
Samson Zhou\thanks{Texas A\&M University. \email{samsonzhou@gmail.com}. 
Samson Zhou is supported in part by NSF CCF-2335411 and gratefully acknowledges funding provided by the Oak Ridge Associated Universities (ORAU) Ralph E. Powe Junior Faculty Enhancement Award.}
}
\maketitle

\begin{abstract}
Perfect $L_2$ sampling is a powerful primitive for turnstile data streams, allowing algorithms to extract an index with probability proportional to its squared frequency, with applications to dynamic join size estimation, query optimization, heavy-hitters, sampling from self-joins, and real-time monitoring of evolving datasets. The setting $p=2$ is also the largest $p$ for which $L_p$ sampling can be achieved in polylogarithmic space. However, existing space-optimal perfect samplers rely on explicit polynomial duplication, requiring inefficient polynomial update times per stream arrival. Resolving a open question left by Swartworth, Woodruff, and Zhou [FOCS 2025], we present a perfect $L_2$ sampler that achieves an expected $\text{polylog}(n)$ update time while maintaining an optimal $\tilde{\mathcal{O}}(\log^3 n)$ bits of space.
\end{abstract}

\section{Introduction}
The scale of modern data systems often makes it infeasible to store or repeatedly process an entire dataset. 
In many settings, such as network monitoring and large-scale database systems, data instead arrives as a continuous stream of updates. 
The \emph{streaming model of computation} formalizes this setting by requiring algorithms to process the input using memory far smaller than the size of the data itself. 
In the standard turnstile model, the stream implicitly defines a high-dimensional frequency vector $x\in\mathbb{R}^n$ that is initialized to the all zeros vector and evolves through updates of the form $x_i\gets x_i+\Delta$, where $\Delta$ may be positive or negative. 
This model naturally captures dynamic settings in which data is both inserted and deleted, such as updates to database records or changes in network traffic.

A central goal in the streaming model is to design primitives that extract useful information from the frequency vector while respecting stringent space constraints. 
One particularly powerful primitive is $L_p$ sampling, which returns an index $i\in[n]$ with probability proportional to $(1\pm\eps)\cdot\frac{|x_i|^p}{\|x\|_p^p}+\frac{1}{\poly(n)}$. 
Informally, an $L_p$ sampler selects coordinates according to their contribution to the $p$-th frequency moment, biasing the sample toward larger entries as $p$ increases.
The sampler is called \emph{perfect} if $\eps=0$ and \emph{approximate} otherwise. 
Because $L_p$ sampling captures structure in the frequency vector beyond global aggregates by extracting representative coordinates according to the geometry induced by the $L_p$ norm, $L_p$ sampling has emerged as a powerful primitive in a variety of sublinear and streaming algorithms. 
In particular, it has found applications in a broad spectrum of settings, including network monitoring and anomaly detection~\cite{EstanV03,HuangNGHJJT07,GilbertKMS01,ThottanLJ10,MaiCSYZ06}, distributed computation frameworks~\cite{CormodeMYZ12,TirthapuraW11,JayaramSTW19,WoodruffZ12,CormodeF14}, database query processing tasks~\cite{HaasNSS96,LiptonNS90,GibbonsM98,CohenG19,Olken93,Haas16,LiptonN95,HaasS92}, and techniques for compact representation and summarization of large datasets~\cite{FriezeKV04,DeshpandeRVW06,DeshpandeV07,AggarwalDK09,IndykMMM14,MahabadiIGR19,Cohen2020wor,IndykMGR20,MahabadiRWZ20,JayaramWZ22,MahabadiWZ22,GollapudiMS23,MahabadiT23,LinSWXZ26}. 
Despite this broad applicability, the complexity and behavior of $L_p$ sampling depend strongly on the choice of $p$, and different regimes require fundamentally different algorithmic ideas.

Among the family of $L_p$ samplers, the case $p=2$ occupies a particularly important position. 
From a geometric perspective, the $L_2$ norm underlies Euclidean structure and is therefore deeply connected to many algorithmic techniques based on linear algebra and sketching. 
Within database systems, the second frequency moment $F_2=\|x\|_2^2$ corresponds directly to self-join size, making $L_2$ sampling a natural primitive for approximate query processing, join estimation, and regression over joins without explicitly materializing them, where sampling-based techniques are commonly used~\cite{JayaramSWY21}.  
Moreover, it is used in a number of subroutines such as streaming algorithms for $L_p$ norm estimation for $p>2$ \cite{WoodruffZ21,WoodruffXZ25}. 
Thus, progress on $L_2$ sampling directly influences the complexity of a much broader class of problems.
In particular, $p=2$ is the largest exponent for which polylogarithmic space $L_p$ sampling is possible, as $\tilde{\Omega}(n^{1-2/p})$ space is necessary for $p>2$~\cite{GribelyukLWYZ25,WoodruffXZ25,JiangLY26}. 

\paragraph{Previous work.}
Despite its central role, constructing efficient $L_p$ samplers has been well-studied in the streaming literature.  
In the insertion-only setting, the classical reservoir sampling technique~\cite{Vitter85} gives an exact $L_1$ sampler using $\O{\log n}$ bits of space with no distortion, both multiplicative and additive. 
However, beyond the $L_1$ case or outside insertion-only streams, the problem becomes significantly more complex and the question of whether such samplers exist was explicitly posed by \cite{CormodeMR05}, and subsequently resolved in the affirmative for $p \le 2$ in the approximate setting~\cite{MonemizadehW10}. 
These initial constructions used $\poly\left(\frac{1}{\eps}, \log n\right)$ space, and were later improved to near-optimal bounds by \cite{AndoniKO11} and by \cite{JowhariST11}, achieving roughly $\O{\frac{1}{\eps^{\max(1,p)}} \log^2 n}$ space for $p \in (0,2)$ and $\O{\frac{1}{\eps^2} \log^3 n}$ for $p=2$. 
While these results established near-optimal space complexity for approximate sampling, they left open finer questions regarding distortion and randomness requirements. 
This gap was later resolved by \cite{JayaramW18}, who introduced the first perfect $L_p$ samplers for $p \in (0,2)$ with $\tO{\log^2 n}$ space and for $p=2$ with $\tO{\log^3 n}$ space.

\paragraph{Update time.}
Despite resolving the space complexity, the algorithms of \cite{JayaramW18} and related work~\cite{CohenG19} suffer from a severe computational bottleneck, duplicating each coordinate $n^C$ times and assigning each copy an independent exponential weight. 
Consequently, a single update requires $\Omega(n^C)$ structural operations, so that the update time per element is polynomial in $n$, which can be prohibitively large in many settings. 
This sparked work on fast perfect $L_p$ sampling. 
For insertion-only streams, \cite{PettieW25} achieved $\O{1}$ update time for $p \le 1$ in the random oracle model. 
Subsequently, \cite{SwartworthWZ25} introduced a continuous simulation framework for general turnstile streams, achieving polylogarithmic update time for $p<2$. 
Unfortunately, the framework of \cite{SwartworthWZ25} does not extend to the critical case $p=2$, where the limiting behavior underlying the continuous simulation approach changes qualitatively and the techniques used for $p<2$ no longer apply (see \Cref{sec:overview} for more details). 

At first glance, this distinction may appear to be merely a change of regime. 
However, the situation is less clear from a complexity perspective. 
Indeed, while a combination of techniques from \cite{JiangLY26,SwartworthW26} showed an $\Omega(\log^3 n)$ space lower bound for perfect $L_2$ sampling, perfect samplers for $p<2$ are known using only $\tO{\log^2 n}$ space~\cite{JayaramW18,SwartworthWZ25}. 
Thus, it is natural to ask whether the additional memory complexity of the $L_2$ regime fundamentally obstructs fast algorithms, or whether one can still achieve polylogarithmic update time while maintaining optimal space.

\paragraph{Our contributions.}
In this paper, we resolve the open question of whether perfect $L_2$ sampling admits both optimal space and fast update time in the turnstile model. 
We introduce a new approach for implementing perfect $L_2$ sampling without relying on the explicit polynomial duplication used in previous constructions to achieve the first perfect $L_2$ sampler with expected polylogarithmic update time and optimal space complexity.

\begin{restatable}{theorem}{thmmain}
\label{thm:main}
Given a sequence of turnstile updates defining a frequency vector $x\in\mathbb{R}^n$, and a failure probability parameter $\delta\in(0,1)$, there exists a perfect $L_2$ sampler that outputs a failure symbol $\bot$ with probability at most $\delta$. 
Conditioned on returning an index $i^*\in[n]$, 
\[\PPr{i^*=j}=\frac{x_j^2}{\|x\|_2^2}\pm\frac{1}{\poly(n)}.\]
The algorithm uses $\tO{\log^3 n\log\frac{1}{\delta}}$ bits of space and processes each stream update in $\polylog(n)$ time in expectation.
\end{restatable}
We remark that the output $\bot$ corresponds to a failure event that occurs with probability at most $\delta$; this is standard in perfect sampling constructions and can be reduced arbitrarily by repetition without affecting the conditional sampling guarantee.
Conceptually, our result shows that the additional complexity of the $L_2$ regime does not fundamentally obstruct efficient algorithms. 
Although the methods developed for $p<2$ do not directly extend to $p=2$, we show that this qualitative change in behavior need not incur an algorithmic penalty. 
Thus, despite occupying a distinct space complexity regime, perfect $L_2$ sampling remains achievable with both optimal space and efficient updates.

\section{Motivation and Applications}

\paragraph{Statistical indistinguishability.}
Sampling is a fundamental primitive for extracting information from massive datasets. 
In particular, $L_2$ sampling has been used as a subroutine for central streaming tasks such as heavy-hitters, norm and moment estimation, duplicate detection, and data summarization~\cite{MonemizadehW10,AndoniKO11,JowhariST11,BravermanOZ12,JayaramW18}. 
Since these summaries effectively represent the dataset itself, it is important that the samples accurately reflect the target distribution.

Approximate samplers permit multiplicative distortion in the sampling probabilities. 
While this distortion may be small for an individual sample, the effect can accumulate when many samples are repeatedly combined or reused throughout a larger computational pipeline. 
Consequently, hidden statistical bias may propagate into downstream procedures and alter their behavior. 
Perfect $L_2$ samplers avoid this issue by reducing the additive distortion in the output distribution to an arbitrarily small inverse-polynomial quantity while preserving essentially the same asymptotic space complexity. 
For example, if an application extracts $n^c$ samples for some constant $c>0$, then choosing an additive error of $n^{-(c+100)}$ ensures that the resulting joint distribution remains statistically indistinguishable from one obtained using truly exact samples. 
This makes perfect samplers particularly attractive as reusable black-box primitives.

\paragraph{Distributed databases and query processing.}
The distinction between approximate and perfect sampling becomes especially important in distributed database systems. 
Modern data processing architectures partition datasets across many machines, where local summaries are constructed independently and later combined into global estimates. 
Sampling provides a natural mechanism for constructing such summaries while avoiding communication of the raw data itself.

However, local inaccuracies need not remain local. 
Even small distortions introduced by approximate samplers can accumulate when summaries are merged, reused, or incorporated into later stages of processing. 
This phenomenon is particularly relevant in relational settings involving repeated statistical estimation, joins, and self-join computations, where the output of one randomized procedure often becomes the input to another. 
Perfect samplers substantially reduce such accumulated bias and therefore preserve the fidelity of both local and global statistical summaries.

\paragraph{Information leakage and adaptive interaction.}
Approximate sampling probabilities may also reveal unintended information about the underlying dataset. 
Since the permitted multiplicative distortion can depend on global properties of the data, repeated observations of samples may allow an observer to infer information beyond the sampled index itself. 
For instance, suppose a sampler increases the total probability mass of a large set of coordinates by a factor of $(1+\eps)$ whenever a property $\calP$ holds and decreases it by $(1-\eps)$ otherwise. 
Then an observer could estimate whether $\calP$ holds simply by collecting sufficiently many samples and measuring aggregate frequencies.

Perfect $L_2$ samplers mitigate this issue because their output distributions differ from the ideal distribution only by an arbitrarily small additive amount. 
As a result, repeated interactions with the sampler reveal essentially no additional information beyond that inherent to the target distribution itself.

\section{Algorithmic and Technical Overview}
\label{sec:overview}
In this section, we provide an overview of our algorithmic framework and the technical novelties required to achieve polylogarithmic update time for perfect $L_2$ sampling, c.f., \Cref{fig:algorithm_overview}. 
To contextualize our approach, it is instructive to first review the mechanics and limitations of previous sampling frameworks before explaining how we overcome them.

\paragraph{Exponential random variables and duplication.}
The standard approach~\cite{JayaramW18,SwartworthWZ25,WoodruffXZ25} for perfect $L_p$ sampling involves scaling each stream coordinate $x_i$ by an inverse exponential random variable $e_i^{-1/p}$. 
Due to the max-stability property of exponential random variables~\cite{AndoniKO11}, the probability that coordinate $i\in[n]$ yields the largest scaled value $\frac{|x_i|}{e_i^{1/p}}$ is exactly $\frac{|x_i|^p}{\|x\|_p^p}$. 
To recover this maximum element in a data stream, the coordinate is identified using a heavy-hitters data structure such as $\CountSketch$, coupled with a statistical test to ensure the maximum is reliably isolated from the background noise. 
However, the probability of passing this test is heavily correlated with the identity of the maximum element itself, distorting the sampling probabilities by an $\Omega(1)$ multiplicative bias that destroys the ``perfect'' sampling guarantee. 

To handle this issue, \cite{JayaramW18} duplicated each coordinate $i\in[n]$ a total of $k = n^C$ times, scaling each duplicate by its own independent exponential random variable. 
This flattens the background noise so that the statistical test becomes nearly independent of the identity of the maximum element, recovering the desired perfect sampling distribution. 
Unfortunately, explicitly routing $n^C$ duplicate coordinates into the sketch requires a prohibitively high update time of $\Omega(n^C)$.

\begin{figure}[htbp]
\centering
\begin{mdframed}
\textbf{Algorithm Pipeline: Fast Perfect $L_2$ Sampling}
\begin{enumerate}
\item 
\textbf{Explicit Head:} Conceptually duplicate each update $n^C$ times, but explicitly generate only the top $\tau = \Theta(\log n)$ exponential weights.
\item 
\textbf{Tail Simulation:} 
Compress the remaining copies into a variance $V_{\mathrm{tail}}$, sampled via fast quadrature of its $1$-stable continuous CDF.
\item 
\textbf{Sketch Injection:} 
Update a dense Gaussian $\CountSketch$, injecting the head directly and the aggregate tail as a single $\calN(0, V_{\mathrm{tail}})$.
\item 
\textbf{Recovery:}
Run a statistical test and extract the maximum coordinate if it succeeds.
\end{enumerate}
\end{mdframed}
\vspace{-0.75em}
\caption{Overview of the framework, reducing update time to an expected $\poly(\log n)$.}
\label{fig:algorithm_overview}
\vspace{-1em}
\end{figure}

\paragraph{Algorithmic framework: explicit heads and simulated tails.}
Our goal is to preserve the exact statistical behavior of this duplication process while avoiding its explicit realization. 
Our approach is based on a simple structural observation: although the duplicated coordinates are numerous, only a small number contribute significantly to the sketch on any given update. 
We therefore separate the duplicated copies into two components. 
For a coordinate $i \in [n]$, the largest $\tau = \Theta(\log n)$ duplicated coordinates are generated and maintained explicitly (the ``head''), while the remaining $m = k - \tau$ coordinates are aggregated into a single random contribution (the ``tail'').

Crucially, this decomposition integrates seamlessly into the sketching framework. 
Following the approach of \cite{SwartworthWZ25}, we utilize a dense Gaussian $\CountSketch$. 
Because Gaussian random variables are exactly $2$-stable, the aggregate contribution of the uninstantiated tail duplicates in each sketch row behaves identically to a single Gaussian variable $\calN(0, V_{\mathrm{tail}})$. 
Rather than explicitly materializing the $m$ tail coordinates, we can compress their contribution into a single sampled variance parameter $V_{\mathrm{tail}}$. 
Thus, as detailed in \Cref{alg:update}, the algorithm maintains only the explicit heavy coordinates $\calH_i$, while the remaining copies are implicitly injected into the sketch through the sampled tail variance.

\begin{algorithm}[!htb]
\caption{Perfect $L_2$ stream update processing}
\label{alg:update}
\begin{algorithmic}[1]
\Require{Stream update $(i,\Delta)$}
\Ensure{Update the dense $\CountSketch$ associated with coordinate $i$}
\State{$r=\O{\log^2 n}$, $\tau=\O{\log n}$, $k=n^C$}
\State{Generate $(\calH_i,V_{\mathrm{tail}})$ using \Cref{alg:l2:generator}}
\For{each row $\ell\in[R]$}
\For{each $(j,Y_{(j)})\in\calH_i$}
\State{Draw $g_{\ell,i,j}\sim\calN(0,1)$}
\State{Update row $\ell$ by $\Delta\cdot g_{\ell,i,j}\cdot\sqrt{Y_{(j)}}$}
\EndFor
\State{Draw $g_{\ell,i,\mathrm{tail}}\sim\calN(0,1)$}
\State{Update row $\ell$ by $\Delta\cdot g_{\ell,i,\tail}\cdot\sqrt{V_{\mathrm{tail}}}$}
\EndFor
\end{algorithmic}
\end{algorithm}

The explicit heavy coordinates $\calH_i$ and the necessary truncation parameters can be generated efficiently using standard properties of exponential spacings, as shown in \Cref{alg:l2:generator}. 

\begin{algorithm}[!htb]
\caption{Generation of heavy coordinates and tail variance}
\label{alg:l2:generator}
\begin{algorithmic}[1]
\Require{Coordinate $i$, duplication count $k$, truncation parameter $\tau$ from \Cref{alg:update}}
\Ensure{$(\calH_i,V_{\mathrm{tail}})$}
\State{Seed a pseudorandom generator using the instance ID and coordinate $i$}
\State{Generate $\delta_1,\ldots,\delta_\tau\sim\Exp(1)$}
\State{Generate order statistics}
\[E_{D(j)}=E_{D(j-1)}+\frac{\delta_j}{k-j+1},\qquad E_{D(0)}=0\]
\State{Set $\calH_i=\left\{\left(j,\frac{1}{E_{D(j)}}\right)\right\}_{j\in[\tau]}$ and draw $U\sim\Unif(0,1)$}
\State{$m\gets k-\tau$ and $R\gets\frac{1}{E_{D(\tau)}\cdot m}$}
\State{Compute $z^*$ via inverse CDF using \Cref{alg:oracle:invert}}
\State{Set $V_{\mathrm{tail}}\gets \max\left(0,m(z^*+\ln m)\right)$}
\State{Return $(\calH_i,V_{\mathrm{tail}})$}
\end{algorithmic}
\end{algorithm}

\paragraph{Sampling the continuous tail distribution and the $L_2$ barrier.}
The core algorithmic challenge lies in generating the variance parameter $V_{\mathrm{tail}}$. 
We simulate the sum of the tail variables $V = \sum_{j=1}^m e_{i,j}^{-1}$ as a single random draw from its limiting continuous distribution, $\calZ$.
We perform this via inverse transform sampling (\Cref{alg:oracle:invert}): we draw $U \sim \Unif(0,1)$ and compute $z^*$ such that $F_{\calZ}(z^*) \approx U$, where $F_{\calZ}$ is the Cumulative Distribution Function (CDF) of the limiting tail.

\begin{algorithm}[!htb]
\caption{Inverse CDF via binary search}
\label{alg:oracle:invert}
\begin{algorithmic}[1]
\Require{$U \in (0,1)$, interval $[\ell,u]=[-\O{\log n}, \O{\log n}]$}
\Ensure{$z^*$ such that $F_{\mathcal Z}(z^*) \approx U$}
\While{$u-\ell > \frac{1}{n^{\O{C+c}}}$}
\State $m\gets\frac{\ell+u}{2}$
\If{$\mathrm{CDF}(m) < U$}
\Comment{$\mathrm{CDF}(x)$ denotes $F_{\mathcal Z}(x)$ evaluated via \Cref{alg:oracle}}
\State{$\ell \gets m$}
\Else
\State{$u \gets m$}
\EndIf
\EndWhile
\State{Draw $W\sim\Unif(0,1)$}
\State{\Return $\ell+W(u-\ell)$}
\end{algorithmic}
\end{algorithm}

To evaluate $F_{\calZ}(x)$ efficiently inside the binary search, we must rely on numerical inversion of the Gil-Pelaez integral. This requires obtaining the characteristic function of the tail distribution. 

This introduces the main technical barrier distinguishing perfect $L_2$ sampling from prior work. 
For $p < 2$~\cite{SwartworthWZ25}, the sum of the tail elements $V = \sum_{j=1}^m e_{i,j}^{-2/p}$ has a heavy tail that converges tightly to a $\frac{p}{2}$-stable distribution. 
Because $\frac{p}{2} < 1$, its characteristic function is well-behaved near the origin, allowing the CDF to be evaluated directly using standard quadrature in polylogarithmic time.
However, for the $p = 2$ regime, the variables $e_{i,j}^{-1}$ are inverse standard exponentials, which have infinite expectation. 
Consequently, as the number of tail items $m$ grows, the sum $V$ diverges at a rate of $m\ln m$ rather than converging to a standard stable limit. 
This is not merely a different analytic computation; the logarithmic divergence causes the limiting distribution to be exactly $1$-stable, where the characteristic function develops a non-integrable logarithmic singularity near the origin in the Fourier domain. 
Because standard numerical integration techniques require analytic or sufficiently smooth integrands, the singularity at this $1$-stable boundary causes classical numerical quadrature to either require exponential time to converge or fail outright.

\paragraph{Recentering and delaying the limit.}
To stabilize this divergence, we recenter the sum by defining the random variable $Z_m=\frac{1}{m}\sum_{j=1}^m e_{i,j}^{-1} - \ln m$.
Unlike the $p<2$ regime, where the limiting characteristic function possesses an integrable singularity that evaluates cleanly to an incomplete Gamma function, the $p=2$ boundary results in a $1$-stable limit with a non-integrable logarithmic singularity ($\sim 1/u$) at the origin in the Fourier domain. 
Na\"{i}vely recentering the random variable and taking the continuous limit fails because this non-integrable singularity prevents the exchange of the limit and the integral, breaking the standard numerical quadrature required for fast evaluation.

We overcome this barrier by \emph{delaying the limit} and performing analytic surgery at a finite tail size $m$. 
By considering the characteristic function $\Psi(t) = \Ex{e^{itZ_m}}$ for finite $m$, we observe that we can express its core integral as:
\[\int (e^{itu}-1)u^{-2}e^{-1/(um)}\,du=\int[(e^{itu}-1-itu)+itu]u^{-2}e^{-1/(um)}\,du.\]
By explicitly isolating the singular linear term $itu$, the extracted component $it\int \frac{1}{u}e^{-1/(um)}\,du$ evaluates exactly to an Exponential Integral ($E_1$). 
Using its asymptotic expansion via the Entire Exponential Integral ($\mathrm{Ein}$), $E_1(z)=-\gamma-\ln z+\mathrm{Ein}(z)$, it naturally generates a $+it\ln m$ phase shift. 
This algebraically annihilates the algorithmic $-it\ln m$ recentering penalty \emph{before} taking $m \to \infty$. 
Thus, delaying the limit completely removes the logarithmic divergence, leaving a smooth, bounded remainder defined by Sine and Cosine Integrals that permits fast Gauss-Legendre quadrature.

\paragraph{Trigonometric integrals and fast quadrature.}
After cancellation, the remaining regularized expression becomes $\int_0^R\frac{e^{itu}-1-itu}{u^2}\,du$, where $R$ is the truncation parameter for the magnitude of the tail elements. 
This differs substantially from the $p<2$ setting, which resolves into incomplete Gamma functions. 
Instead, applying integration by parts transforms the expression into $\frac{e^{itu}-1}{u}$, yielding real and imaginary parts $\int\frac{\cos(tu)-1}{u}\,du$ and $\int\frac{\sin(tu)}{u}\,du$.
These correspond precisely to the Cosine Integral $\mathrm{Cin}$ and Sine Integral $\mathrm{Si}$. 
Consequently, we obtain an explicit analytic form for the characteristic function at the $1$-stable boundary:
\[\Psi(t)=\exp\left(it(\ln R-\gamma)-\frac{e^{itR}-1-itR}{R}-it\,\mathrm{Cin}(tR)-t\,\mathrm{Si}(tR)\right).\]

With the explicit characteristic function derived, we can evaluate its CDF through the Gil-Pelaez inversion formula.
To achieve polylogarithmic update time, the infinite integration domain must be truncated to a logarithmic interval while introducing only negligible error. 
The magnitude of the characteristic function is governed by $\Re(\ln \Psi(t))=\frac{1-\cos(tR)}{R}-t\,\mathrm{Si}(tR)$. 
We show that the Sine Integral is the dominating term, yielding the global bound $|\Psi(t)|\le\exp\left(\frac{2}{R}-\kappa |t|\right)$ for a universal constant $\kappa>0$.
Using standard concentration and anti-concentration properties (c.f., R\'enyi's representation theorem for exponential spacings), we show the algorithmic truncation parameter satisfies $\frac{1}{R}=\Theta(\log n)$. 
This exponential decay guarantees that we can truncate the inversion integral to an interval of size $\O{\log n}$ while incurring only an additive error of $\frac{1}{\poly(n)}$ in the sampling distribution.

Finally, before applying numerical integration, we note that the standard Gil-Pelaez formula contains an imaginary-part operator. This introduces non-analytic behavior that prevents the direct use of exponentially convergent quadrature methods.
To circumvent this, we rewrite the integrand as $g(t)=\frac{e^{itx}\overline{\Psi(t)}-e^{-itx}\Psi(t)}{2it}$ and show that $g(t)$ extends holomorphically into the complex plane, allowing the integration contour to be enclosed within a Bernstein ellipse.
Because the domain integral is truncated at $R=\Theta\left(\frac{1}{\log n}\right)$ and the inversion integral at $t=\Theta(\log n)$, then $|tR|=\Theta(1)$. 
Thus, the trigonometric integrals do not exhibit rapid growth, enabling exponentially convergent Gauss-Legendre quadrature to evaluate the CDF using only $\O{\log^2 n}$ points.
We summarize the main differences between our approach and previous works in \Cref{tab:comparison}. 

\begin{table*}[htbp]
\centering
\renewcommand{\arraystretch}{1.1}
\setlength{\tabcolsep}{20pt} 
\begin{tabular}{@{} l l l @{}}
\toprule
\textbf{Property} & \textbf{Prior Work~\cite{SwartworthWZ25}} ($p < 2$) & \textbf{Our Work} ($p = 2$) \\
\midrule
\textbf{Tail Limit} & Convergent $(p/2)$-stable & Diverges $m \ln m$ ($1$-stable) \\
\textbf{Recentering?} & No & Yes ($-it\ln m$ shift penalty) \\
\textbf{Singularity} & Integrable near origin & Non-integrable log ($\sim 1/u$) \\
\textbf{Limit Eval.} & Direct continuous limit & Delayed limit \\
\textbf{Integral Form} & Incomplete Gamma ($\gamma$) & Exp/Trig ($E_1, \mathrm{Ein}, \mathrm{Si}, \mathrm{Cin}$) \\
\bottomrule
\end{tabular}
\vspace{0.75em}
\caption{The $p=2$ boundary introduces non-integrable singularities, breaking prior techniques and requiring a fundamentally different analytic toolkit.}
\label{tab:comparison}
\vspace{-1.5em}
\end{table*}

\section{Preliminaries}
\label{sec:preliminaries}
In this section, we briefly recall standard concepts that we will use in our algorithmic design and analysis. 

Given an integer $n>0$, we use $[n]$ to denote the set $\{1,2,\ldots,n\}$. 
Recall that the $L_2$ norm of a vector $x\in\mathbb{R}^n$ is defined by $\sqrt{x_1^2+\ldots+x_n^2}$ while the $F_2$ moment of $x$ is defined by $\|x\|_2^2$. 
We use $\poly(n)$ to denote an arbitrary but fixed polynomial in $n$, whose degree can be adjusted by appropriately adjusting corresponding constants. 
We use $\polylog(n)$ to denote $\poly(\log n)$. 
We say an event happens with high probability if it occurs with probability $1-\frac{1}{\poly(n)}$.

\begin{theorem}[Markov's Inequality]
For any non-negative random variable $X$ and any constant $t > 0$, we have $\PPr{X\ge t}\le\frac{\Ex{X}}{t}$. 
\end{theorem}

\subsection{Probability Distributions}

\begin{definition}[Exponential random variable]
A continuous random variable $E$ follows an exponential distribution with scale $\lambda>0$, denoted $E \sim \Exp(\lambda)$, if its probability density function is given by $f_E(x) = \lambda e^{-\lambda x}$ for $x \ge 0$, and $0$ otherwise. 
\end{definition}

\begin{fact}
\label{fact:exp}
The exponential distribution satisfies the properties:
\begin{enumerate}
\item 
(Memory-less property) 
For all $s,t \ge 0$, $\PPr{E>s+t\,\mid\,E>s}=\PPr{E>t}$. 
\item 
(Max-stability) 
If $E_1,\ldots,E_n$ are independent exponential random variables with rates $\lambda_1,\ldots,\lambda_n$, then $\max_i E_i$ has cumulative distribution function
\[\PPr{\max_i E_i \le t}=\prod_{i=1}^n \left(1-e^{-\lambda_i t}\right).\]
\end{enumerate}
\end{fact}

\begin{definition}[Gaussian random variable]
A continuous random variable $g$ follows a standard Gaussian (or normal) distribution, denoted $g \sim \mathcal{N}(0,1)$, if its probability density function is $f_g(x) = \frac{1}{\sqrt{2\pi}} e^{-x^2/2}$. 
\end{definition}

\begin{definition}[Cauchy random variable]
A continuous random variable $C$ follows a standard Cauchy distribution if its probability density function is $f_C(x) = \frac{1}{\pi(1+x^2)}$. 
\end{definition}

\begin{definition}[$p$-stable]
A distribution $\mathcal{D}$ over $\mathbb{R}$ is called \emph{$p$-stable} for $0 < p \le 2$ if for any independent random variables $X_1,\ldots,X_k \sim \mathcal{D}$ and any constants $c_1,\dots,c_k \in \mathbb{R}$, we have
\[\sum_{j=1}^k c_j X_j\sim\left(\sum_{j=1}^k |c_j|^p\right)^{1/p} X,\]
where $X \sim \mathcal{D}$, so that $\sim$ denotes equality in distribution.
\end{definition}

\begin{fact}
The Cauchy distribution is $1$-stable, while the Gaussian distribution is $2$-stable. 
Moreover, if $g_1,g_2$ are independent standard Gaussian random variables, then $\frac{g_1}{g_2}$ follows the standard Cauchy distribution.
\end{fact}

\subsection{Characteristic Functions and Inversion Formulas}
To analyze the continuous limits of the sums of random variable when explicit probability density functions become mathematically intractable, we instead use Fourier-domain transformations.

\begin{definition}[Characteristic function and cumulant-generating function]
The characteristic function of a random variable $X$, denoted $\Psi_X(t)$, is the expected value of $e^{itX}$, representing the Fourier transform of its probability measure. 
For any real $t \in \mathbb{R}$:
\begin{align*}
\Psi_X(t) = \Ex{e^{itX}} = \int_{-\infty}^{\infty} e^{itx} f_X(x) \, dx,
\end{align*}
where $f_X(x)$ is the probability density function of $X$, and $i = \sqrt{-1}$. 
The natural logarithm of the characteristic function, $\ln \Psi_X(t)$, is known as the cumulant-generating function.
\end{definition}
Unlike moment-generating functions, the characteristic function is globally bounded and uniquely exists for all random variables, including those with heavy tails (e.g., Cauchy distributions). 

\begin{theorem}[Gil-Pelaez Inversion Formula]
\label{thm:gil_pelaez}
Let $X$ be a continuous random variable have characteristic function $\Psi_X(t)$. 
Then the cumulative distribution function (CDF) $F_X(x) = \PPr{X \le x}$ can be computed for any real $x$ via the integral:
\begin{align*}
F_X(x) = \frac{1}{2} - \frac{1}{\pi} \int_0^\infty \frac{\Im\left(e^{-itx}\Psi_X(t)\right)}{t} \, dt,
\end{align*}
where $\Im(\cdot)$ denotes the imaginary component of a complex number.
\end{theorem}
A standard approach for generating random samples from a target distribution with continuous and monotonically increasing CDF $F(x)$ is inverse transform sampling, which is simply to generate a uniform random variable $U \sim \Unif(0,1)$ and then compute $X = F^{-1}(U)$. 

\subsection{Statistical Distances}
To quantify the discrepancy between an algorithmic output distribution and the ideal target distribution, we utilize standard statistical distance metrics.

\begin{definition}[Kolmogorov-Smirnov Distance]
The Kolmogorov-Smirnov distance (KS) is the maximum deviation between the Cumulative Distribution Functions $F(x)$ and $G(x)$ of two probability distributions, i.e., $d_{\text{KS}}(F, G) = \sup_{x \in \mathbb{R}} |F(x) - G(x)|$.
\end{definition}

\begin{definition}[Total variation distance]
Let $P$ and $Q$ be probability distributions on $\mathbb{R}$ with probability density functions $p$ and $q$. 
Then the total variation distance is defined by $d_{\mathrm{TV}}(P,Q)=
\frac{1}{2} \int_{\mathbb{R}} |p(x) - q(x)| \, dx$.
\end{definition}

\subsection{Real Analysis and Special Functions}

Evaluating continuous limits mathematically frequently yields classical non-elementary integrals and requires rigorously interchanging limits and integrals.

\begin{theorem}[Dominated Convergence Theorem]
\label{thm:dominated:convergence}
Let $\{f_n(x)\}$ be a sequence of measurable functions converging pointwise to a limit function $f(x)$. 
If there exists an absolutely integrable bounding function $g(x)$ such that $|f_n(x)| \le g(x)$ for all $n$ and $x$, then the limit and integral operators may be exchanged:
\begin{align*}
\lim_{n \to \infty} \int f_n(x) \, dx = \int f(x) \, dx.
\end{align*}
\end{theorem}

\begin{definition}[Special Integral Functions]
\label{def:special:integrals}
The mathematical evaluation of limits frequently isolates the following classical analytic functions:
\begin{align*}
\text{Exponential Integral:} \quad & E_1(z) = \int_z^\infty \frac{e^{-v}}{v} \, dv, \quad (z > 0) \\
\text{Entire Exponential Integral:} \quad & \Ein(z) = \int_0^z \frac{1-e^{-v}}{v} \, dv = \sum_{k=1}^\infty \frac{(-1)^{k-1} z^k}{k \cdot k!} \\
\text{Sine Integral:} \quad & \Si(x) = \int_0^x \frac{\sin(u)}{u} \, du \\
\text{Cosine Integral:} \quad & \Cin(x) = \int_0^x \frac{1 - \cos(u)}{u} \, du
\end{align*}
\end{definition}

\subsection{Order Statistics}

\begin{definition}[Order statistics and anti-ranks]
Given a sample of $k$ independent and identically distributed random variables $X_1, X_2, \ldots, X_k$, their order statistics are the sequence of values sorted in non-decreasing order, so that $X_{D(1)} \le X_{D(2)} \le \ldots \le X_{D(k)}$. 
Here $D(1),\ldots,D(k)$ are called the anti-ranks for $X$. 
\end{definition}

\begin{theorem}[R\'{e}nyi Representation Theorem for Exponential Spacings]
\label{thm:renyi}
\cite{Nagaraja06}
Let $E_{(1)} \le E_{(2)} \le \dots \le E_{(k)}$ be the order statistics of $k$ independent standard exponential random variables $E_i \sim \Exp(1)$. 
The $j$-th order statistic is equal in distribution to 
\begin{align*}
E_{D(j)} \sim \sum_{\ell=1}^j \frac{\delta_\ell}{k - \ell + 1},
\end{align*}
i.e., a linear combination of independent standard exponentials $\delta_1, \ldots, \delta_j \sim \Exp(1)$. 
\end{theorem}

\begin{lemma}[Lemma 2 in \cite{JayaramW18}]
\label{lem:hidden:exps}
Let $x\in\mathbb{R}^n$ and define $z_{i,j}=\frac{x_i}{e_{i,j}^{1/p}}$, where the variables $e_{i,j}$ are independent exponential random variables and $j\in[n^c]$. 
Set $N=n^{c+1}$, and let $z\in\mathbb{R}^{N}$ denote a flattened vector obtained from the collection of values $\{z_{i,j}\}$. 
Similarly, define $F\in\mathbb{R}^{N}$ to be the duplicated vector obtained by repeating each coordinate of $x$ exactly $n^c$ times. 
Then for every integer $k\in[1,N-n^{9\zeta/10})$, with probability at least $1-\O{e^{-n^{c/3}}}$
we have
\[\left|z_{D(k)}\right|=\left[\left(1\pm\O{n^{-\zeta/10}}\right)\cdot\sum_{\tau=1}^{k}\frac{E_{\tau}}{\EEx{}{\sum_{j=\tau}^{N}|F_{D(j)}|^p}}\right]^{-1/p}.
\]
\end{lemma}

\subsection{Streaming Algorithms}
We use a dense variant of $\CountSketch$ as the basic estimation primitive throughout our construction, described in \Cref{alg:dense:countsketch}. 
Unlike sparse hashing-based sketches, dense $\CountSketch$ projects the frequency vector onto independent Gaussian directions, allowing each coordinate estimate to be expressed as an inner product with a corresponding Gaussian column. 
This representation yields uniform additive error guarantees across all coordinates, with the error proportional to the global $L_2$ norm of the stream. 
The dense structure is particularly convenient in our setting because the Gaussian random variables interacts naturally with the tail noise introduced by the perfect sampling procedure.

\begin{algorithm}[!htb]
\caption{Dense $\CountSketch$}
\label{alg:dense:countsketch}
\begin{algorithmic}[1]
\Require{Frequency vector $x\in\mathbb{R}^n$}
\Ensure{Estimates of $x_i$ for all $i\in[n]$}
\State{Let $G\in\mathbb{R}^{R\times n}$ be a matrix where each entry is drawn from $\calN(0,1)$}
\State{Let $G_i$ denote the $i$-th column of $G$}
\State{\Return $\widehat{x_i}=\frac{1}{R}\cdot\langle G_i,Gx\rangle$ as the estimate for $x_i$}
\end{algorithmic}
\end{algorithm}
Recall the following guarantees for dense $\CountSketch$ in \Cref{alg:dense:countsketch}, e.g., from \cite{SwartworthWZ25}, by taking the median of $\O{\log n}$ groups of tables with $r$ rows:
\begin{theorem}
\label{thm:dense:countsketch}
\cite{SwartworthWZ25}
With probability at least $1-\frac{1}{\poly(n)}$, the dense $\CountSketch$ data structure with $\O{r\log n}$ rows returns estimates $\widehat{x_1},\ldots,\widehat{x_n}$ satisfying
\[|\widehat{x_i}-x_i|\le \frac{4}{\sqrt{r}}\|x\|_2\]
for every $i\in[n]$. 
The total space complexity is $\O{r\log^2 n}$ bits.
\end{theorem}

\section{Analysis of the Continuous Tail Distribution}
\label{sec:tail}
\Cref{alg:oracle:invert} replaces the explicit contribution of the $m=k-\tau$ tail coordinates with a single sample from a limiting continuous distribution $\mathcal{Z}$. 
In this section, we formalize this limiting object and derive the normalization needed to preserve the induced variance distribution. 
The main challenge is that the reciprocal exponential variables governing the tail are heavy-tailed: before truncation they have infinite expectation, and even after conditioning on the $\tau$-th order statistic their moments exhibit logarithmic growth, scaling as $\O{\log m}$.
As a result, standard concentration arguments and classical central limit heuristics do not apply directly.

Our analysis proceeds in three steps. 
First, we characterize the conditioning induced by removing the largest $\tau$ coordinates and show that the remaining variables admit a simple truncated inverse-exponential description. 
Second, we identify the natural normalization under which the aggregate tail contribution converges to a limiting continuous distribution. 
Unlike Gaussian limits arising from finite-variance settings, the logarithmic divergence of the reciprocal exponential variables leads to a $1$-stable phenomenon, requiring a centering correction of $\ln m$. 
Finally, we derive an explicit representation of the limiting characteristic function and show how it can be evaluated efficiently for use inside the inverse-CDF sampling procedure.

At a high level, the proof may be viewed as an analytic surgery argument. 
Directly evaluating the characteristic function produces a non-integrable $\frac{1}{u}$ singularity near the origin, preventing a straightforward interchange of limits and integration. 
We isolate the singular component, show that it exactly cancels the logarithmic centering term, and then evaluate the remaining regularized expression using special-function identities. 
This converts the limiting law into an explicit form suitable for numerical inversion.

\subsection{Base Variables and Conditioning Shift}

By the memory-less property of the exponential distribution, c.f., \Cref{fact:exp}, conditioning on the $\tau$-th order statistic $E_{D(\tau)}$ implies that the remaining $m$ variables are distributed as independent exponentials shifted by $E_{D(\tau)}$. 
Equivalently, in the inverse representation, the corresponding variables are distributed as independent samples conditioned on lying below a threshold determined by $E_{D(\tau)}$.

Let $X$ denote a random variable distributed as $Y \sim \Exp(1)^{-1}$ conditioned on $Y \le Y_{D(\tau)}$. 
In the algorithmic setting, $Y_{(\tau)} = R \cdot m$. The induced density of $X$ is
\begin{align*}
f_X(x)=\frac{1}{\PPr{Y \le Rm}}x^{-2} e^{-1/x} \mathbf{1}_{(0, Rm]}(x).
\end{align*}
Let $X_1,\dots,X_m$ be i.i.d. copies of $X$. 
We denote the exact discrete tail statistic by $V_{\mathrm{exact}} = \sum_{j=1}^m X_j$. 
To isolate the dominant scaling behavior, we define the normalized statistic
\begin{align*}
Z_m := \frac{1}{m} \sum_{j=1}^m X_j - \ln m.
\end{align*}

\paragraph{The tail-variance functional $V(\cdot)$.}
In \Cref{alg:l2:generator}, the tail-variance functional $V(\cdot)$ maps the inverse-CDF sample $z^*$ produced via \Cref{alg:oracle:invert} back to the variance parameter of the stream. 
This should correspond to the quantity $V_{\mathrm{exact}} = m(z^*+\ln m)$. 
However, since the limiting $1$-stable law underlying $\mathcal{Z}$ is supported on all of $\mathbb{R}$, the sampled value $z^*$ may occasionally satisfy $z^*<-\ln m$, in which case $V_{\mathrm{exact}}$ becomes negative, which cannot correspond to a realizable variance. 
We therefore define the tail variance functional by the truncation 
\[V(z^*):=\max\left(0,\,m(z^*+\ln m)\right),\]
so that the computation of $\sqrt{\widetilde{\sigma_i^2}}$ in \Cref{alg:update} is well-defined and the total variation distance remains within the overall $\pm\nu$ additive error tolerance.

\subsection{Correcting the Order Statistics Boundaries}
\label{sec:order:stats}
We now establish the asymptotic behavior for the tail simulation procedure. 
Specifically, we show that the reciprocal scale parameter $\frac{1}{R}$ is only polylogarithmic in $n$. 
Throughout, let $k=n^C$ for a sufficiently large constant, let the explicitly tracked head size be $\tau=c_1\log n$ for a sufficiently large constant $c_1>0$, so that $m=k-\tau$ is the number of tail coordinates. 
The target additive failure probability is $\nu=n^{-c_2}$ for some input parameter $c_2>0$. 

Informally, the argument works as follows. 
Recall that we define the truncation boundary to be $R=\frac{Y_{(\tau)}}{m}$, where $Y_{(\tau)}$ denotes the $\tau$-th largest inverse exponential. 
Equivalently, if $E_{(1)}\le E_{(2)}\le\ldots$ denote the ordered standard exponentials, then $Y_{(\tau)}=\frac{1}{E_{(\tau)}}$. 
Using R\'enyi's exact representation theorem for exponential order statistics, c.f., \Cref{thm:renyi}, we decompose $E_{(\tau)}=\sum_{j=1}^{\tau}\frac{\delta_j}{k-j+1}$, where the $\delta_j$ are independent exponential random variables $\Exp(1)$. 
Since $\tau=\Theta(\log n)$ and $k=n^C$, then the denominators satisfy $k-j+1=(1-o(1))k$ for all $j\in[\tau]$. 
Thus, 
\[E_{(\tau)} = \frac{1+o(1)}{k}\cdot\sum_{j=1}^{\tau}\delta_j.\]

The sum $\sum_{j=1}^{\tau}\delta_j$ is distributed as a Gamma random variable with mean $\tau$ and satisfies standard Chernoff concentration bounds. 
Therefore, with high probability, $\sum_{j=1}^{\tau}\delta_j =\Theta(\tau)$, 
which implies $E_{(\tau)}=\Theta\left(\frac{\tau}{k}\right)$ and so $Y_{(\tau)}=\Theta\left(\frac{k}{\tau}\right)$. 
Substituting into the definition of $R$ and using $m=k-\tau=(1-o(1))k$, we obtain $R=\Theta\left(\frac{1}{\tau}\right)=\Theta\left(\frac{1}{\log n}\right)$ or equivalently, $\frac{1}{R}=\Theta(\log n)$. 

The crucial observation is that the polynomial duplication scale $k$ cancels completely in the ratio $\frac{Y_{(\tau)}}{m}$ since $m=k-\tau=k-o(k)$. 
Thus, the truncation boundary depends entirely on the logarithmic head parameter $\tau$, rather than by the underlying duplication factor. 

\begin{restatable}[Parameter bounds for the truncation scale]{lemma}{lembounds}
\label{lem:bounds}
Let $Y_{(\tau)}$ denote the $\tau$-th largest order statistic among $k$ independent inverse exponentials $Y_j\sim\Exp(1)^{-1}$. 
Let $k=n^C$, $\tau=c_1\ln n$, and $m=k-\tau$, where $C,c_1>0$ are sufficiently large constants. 
Define the normalized truncation ratio $R=\frac{Y_{(\tau)}}{m}$. 
Then, with probability at least $1-2n^{-c}$, we have $R=\Theta\left(\frac{1}{\log n}\right)$, where the constant $c>0$ may be made arbitrarily large by choosing $c_1$ sufficiently large.
\end{restatable}
\begin{proof}
The proof proceeds by using R\'enyi's representation theorem, c.f., \Cref{thm:renyi} to decouple the standard exponential order statistics and then applying concentration bounds. 
Let $E_1, \ldots, E_k$ be a sequence of independent, identically distributed standard exponential random variables, so that $E_j\sim\Exp(1)$ for all $j\in[k]$. 
Let $E_{D(1)} \le E_{D(2)} \le \ldots \le E_{D(k)}$ denote the corresponding order statistics, so that $Y_j=\frac{1}{E_{D(j)}}$ for all $j\in[k]$. 
By the R\'enyi Representation Theorem, c.f., \Cref{thm:renyi}, the $\tau$-th order statistic of $k$ standard exponential random variables is equal in distribution to a linear combination of independent standard exponentials:
\begin{align*}
E_{D(\tau)} \sim \sum_{j=1}^\tau \frac{\delta_j}{k - j + 1},
\end{align*}
where $\delta_1, \ldots, \delta_\tau \sim \Exp(1)$ are mutually independent standard exponential random variables. 
These hidden exponential random variables are generated by \Cref{alg:l2:generator}. 
Since $k - \tau + 1 \le k - j + 1 \le k$, then we have
\begin{equation} 
\label{eq:sandwich}
\frac{1}{k} \sum_{j=1}^\tau \delta_j \le E_{D(\tau)} \le \frac{1}{k - \tau + 1} \sum_{j=1}^\tau \delta_j.
\end{equation}
Let $S_\tau:=\sum_{j=1}^\tau \delta_j$ be the unscaled sum of the hidden exponentials. 
By the additive property of independent exponentials, $S_\tau$ follows an exact Gamma distribution with shape parameter $\tau$ and rate parameter $1$, denoted $S_\tau \sim \mathrm{Gamma}(\tau, 1)$, so that $\Ex{S_\tau} = \tau$.

To show concentration, we apply a standard Chernoff argument to the sum $S_\tau$. For a standard exponential random variable $\delta_j\sim\Exp(1)$, the moment generating function is
\[\Ex{e^{s\delta_j}} = \int_0^\infty e^{sx}e^{-x}\,dx = \frac{1}{1-s},\]
for $s < 1$. 
By independence, the MGF of $S_\tau$ is $\Ex{e^{s S_\tau}} = (1-s)^{-\tau}$. 
For the upper tail, let $\delta \in (0, 1)$ and choose $s \in (0, 1)$. 
By Markov's inequality, we have
\begin{align*}
\PPr{S_\tau \ge (1+\delta)\tau} = \PPr{e^{s S_\tau} \ge e^{s(1+\delta)\tau}}\le\frac{\Ex{e^{s S_\tau}}}{e^{s(1+\delta)\tau}} = \frac{(1-s)^{-\tau}}{e^{s(1+\delta)\tau}}.
\end{align*}
Observe that for the parameter $s = \frac{\delta}{1+\delta}$, the upper tail bound becomes
\begin{align*}
\PPr{S_\tau \ge (1+\delta)\tau} \le \left( (1+\delta) e^{-\delta} \right)^\tau = \exp\left(-\tau \left(\delta - \ln(1+\delta)\right)\right).
\end{align*}
By similarly applying Markov's inequality to $e^{-s S_\tau}$, we have 
\begin{align*}
\PPr{S_\tau \le (1-\delta)\tau} = \PPr{e^{-s S_\tau} \ge e^{-s(1-\delta)\tau}} \le \frac{\Ex{e^{-s S_\tau}}}{e^{-s(1-\delta)\tau}} = \frac{(1+s)^{-\tau}}{e^{-s(1-\delta)\tau}}.
\end{align*}
By setting $s = \frac{\delta}{1-\delta}$, we have the lower tail bound
\begin{align*}
\PPr{S_\tau \le (1-\delta)\tau} \le \left( (1-\delta) e^{\delta} \right)^\tau = \exp\left(-\tau \left(-\delta - \ln(1-\delta)\right)\right).
\end{align*}
Thus, for the fixed relative deviation of $\delta = 1/2$, the exponent constant evaluates to $\delta - \ln(1+\delta) = 0.5 - \ln(1.5) \approx 0.0945 > 0$ for the upper tail and the exponent constant evaluates to $-\delta - \ln(1-\delta) = -0.5 - \ln(0.5) \approx 0.1931 > 0$ for the lower tail. 
For the upper tail, the exponent constant evaluates to $\delta - \ln(1+\delta) = 0.5 - \ln(1.5) \approx 0.0945 > 0$.
For the lower tail, the exponent constant evaluates to $-\delta - \ln(1-\delta) = -0.5 - \ln(0.5) \approx 0.1931 > 0$.
Thus for $\tau=c_1\ln n$ for sufficiently large $c_1$, we have by a union bound that the probability that $S_\tau$ deviates outside the interval $[0.5\tau, 1.5\tau]$ is at most $2 n^{-c}$. 
Hence, we define the event $\calE$ to be the high-probability event that $0.5\tau \le S_\tau \le 1.5\tau$.

Conditioned on event $\mathcal{E}$, we substitute the bounds for $S_\tau$ back into Equation \ref{eq:sandwich}, obtaining
\begin{align*}
\frac{0.5 \tau}{k} \le E_{D(\tau)} \le \frac{1.5 \tau}{k - \tau + 1}.
\end{align*}
Because $Y_{(\tau)} = \frac{1}{E_{D(\tau)}}$, then it follows that we have 
\begin{align*}
\frac{k - \tau + 1}{1.5 \tau} \le Y_{(\tau)} \le \frac{k}{0.5 \tau}.
\end{align*}
Since we set $R =\frac{Y_{(\tau)}}{m}$ for $m=k-\tau$, then we have:
\begin{align*}
\frac{k - \tau + 1}{1.5 \tau (k-\tau)} \le R \le \frac{k}{0.5 \tau (k-\tau)}.
\end{align*}
Since $(k - \tau + 1) > (k-\tau)$ and $\tau = o(k)$, so that $\tau<0.5k$ for sufficiently large $k=n^c$, then for $\tau = c_1 \ln n$, we have
\begin{align*}
\frac{2}{3c_1 \ln n} < R \le \frac{4}{c_1 \ln n},
\end{align*}
as desired. 
\end{proof}

\subsection{The Continuous Limit and its Characteristic Function}
\label{sec:limit:cf}
We now analyze the limiting distribution describing the simulated tail contribution. 
Recall that the algorithm constructs the centered random variable
\[Z_m=\frac{1}{m}\sum_{j=1}^{m}X_j-\ln m,\]
where the variables $X_1,\ldots,X_m$ are independent and identically distributed with density
\[f_X(x)=\exp\left(\frac{1}{Rm}\right)x^{-2}e^{-1/x},\]
for $x\in(0,Rm]$. 
Our goal is to identify the limiting characteristic function of $Z_m$ as $m\to\infty$.

\begin{restatable}[Limiting characteristic function]{theorem}{thmlimitcf}
\label{thm:limit_cf}
For any parameter $R>0$, as $m\to\infty$, the random variable $Z_m$ converges in distribution to a continuous random variable $\mathcal{Z}$ with characteristic function $\Psi(t)=\exp\left(C_2(t)\right)$, where
\begin{equation}
\label{eq:C2}
C_2(t) = it(\ln R-\gamma) - \frac{e^{itR}-1-itR}{R} - it\Cin(tR) - t\Si(tR).
\end{equation}
Here $\gamma$ denotes the Euler-Mascheroni constant, while $\Si$ and $\Cin$ are the standard sine and cosine integrals, c.f., \Cref{def:special:integrals}.
\end{restatable}
\begin{proof}
Let $\Psi_{Z_m}(t)=\Ex{e^{itZ_m}}$ denote the characteristic function of $Z_m$. 
By independence,
\begin{align*}
\Psi_{Z_m}(t)
&=\Ex{\exp\left(it\left(\frac{1}{m}\sum_{j=1}^{m}X_j-\ln m\right)\right)} \\
&=e^{-it\ln m}\prod_{j=1}^{m}\Ex{\exp\left(\frac{itX_j}{m}\right)} \\
&=e^{-it\ln m}\left(\Ex{\exp\left(\frac{itX}{m}\right)}\right)^m.
\end{align*}
Then by taking logarithms of both sides, 
\[\ln\Psi_{Z_m}(t) = -it\ln m+ m\ln\left(\Ex{\exp\left(\frac{itX}{m}\right)}\right).\]
Let $y_m(t) = \Ex{\exp\left(\frac{itX}{m}\right)}-1$. 
To apply the Taylor series expansion $\ln(1+y_m)=y_m+\O{|y_m|^2}$, we first upper bound $m|y_m(t)|^2$.
Since $|e^{i\theta}-1| \le |\theta|$ for real $\theta$, we have $|y_m(t)| \le \Ex{\left|e^{itX/m}-1\right|} \le \Ex{\left|\frac{tX}{m}\right|} = \frac{|t|}{m}\Ex{X}$.
Using the density of $X$, its expectation is $\int_0^{Rm} x \cdot \exp\left(\frac{1}{Rm}\right) x^{-2} e^{-1/x}\,dx = \exp\left(\frac{1}{Rm}\right) E_1\left(\frac{1}{Rm}\right)$. 
For a fixed $R>0$, $E_1\left(\frac{1}{Rm}\right) = \Theta(\ln m)$ as $m \to \infty$.
Thus $|y_m(t)| = \O{\frac{|t|\ln m}{m}}$, which gives $m|y_m(t)|^2 = \O{\frac{t^2 \ln^2 m}{m}} \to 0$ as $m \to \infty$.
We may therefore expand the logarithm to obtain
\begin{equation}
\label{eq:limit}
\lim_{m\to\infty}\ln\Psi_{Z_m}(t)=\lim_{m\to\infty}\left[m\left(\Ex{\exp\left(\frac{itX}{m}\right)}-1\right)-it\ln m\right].
\end{equation}
We now evaluate the expectation term. 
Using the density of the truncated exponential random variable $X$,
\begin{align*}
\Ex{\exp\left(\frac{itX}{m}\right)}-1&=\int_0^{Rm}\left(e^{itx/m}-1\right)f_X(x)\,dx \\
&=\exp\left(\frac{1}{Rm}\right)\int_0^{Rm}\left(e^{itx/m}-1\right)x^{-2}e^{-1/x}\,dx.
\end{align*}
Substituting $u=\frac{x}{m}$,
\begin{align*}
\Ex{\exp\left(\frac{itX}{m}\right)}-1&=\frac{1}{m}\exp\left(\frac{1}{Rm}\right)\int_0^R\frac{e^{itu}-1}{u^2}e^{-1/(um)}\,du.
\end{align*}
Multiplying by $m$,
\[m\left(\Ex{\exp\left(\frac{itX}{m}\right)}-1\right)=\exp\left(\frac{1}{Rm}\right)\int_0^R\frac{e^{itu}-1}{u^2} e^{-1/(um)}\,du.\]
The integrand has a non-integrable singularity at $u=0$, since $\frac{e^{itu}-1}{u^2}\approx\frac{it}{u}$ near $u=0$.  
To isolate this term, we write
\[e^{itu}-1=\left(e^{itu}-1-itu\right)+itu.\]
Accordingly,
\begin{align*}
\int_0^R\frac{e^{itu}-1}{u^2} e^{-1/(um)}\,du&=\int_0^R\left(\frac{e^{itu}-1-itu}{u^2}\right)e^{-1/(um)}\,du\\
&+ it\int_0^R\frac{1}{u} e^{-1/(um)}\,du.
\end{align*}
We evaluate the second integral by substituting $v=\frac{1}{um}$ and thus $du=-\frac{1}{v^2m}\,dv$, which gives
\begin{align*}
it\int_0^R\frac{1}{u} e^{-1/(um)}\,du &=it\int_{1/(Rm)}^\infty\frac{e^{-v}}{v}\,dv = itE_1\left(\frac{1}{Rm}\right),
\end{align*}
where $E_1$ denotes the exponential integral, c.f., \Cref{def:special:integrals}. 
Using the standard expansion $E_1(z)=-\gamma-\ln z+\Ein(z)$ together with $\Ein(z)=\O{z}$ as $z\to 0$, we have
\begin{align*}
itE_1\left(\frac{1}{Rm}\right)&=it\left(-\gamma+\ln R+\ln m\right)+\O{\frac{t}{Rm}}.
\end{align*}
Substituting this back into \Cref{eq:limit}, the logarithmic term $it\ln m$ cancels the centering shift $-it\ln m$. 
For any fixed $R>0$, the remaining error terms in the expansion of $E_1$ vanish as $m\to\infty$, and the scaling factor $\exp\left(\frac{1}{Rm}\right) \to 1$. 
To evaluate the limit of the first integral, we apply the Dominated Convergence Theorem, c.f., \Cref{thm:dominated:convergence}.  
By Taylor's theorem, the complex exponential satisfies $\left|e^{itu}-1-itu\right| \le \frac{t^2 u^2}{2}$.
Thus, the magnitude of the integrand is uniformly bounded by a continuous function independent of $m$:
\[\left|\frac{e^{itu}-1-itu}{u^2}e^{-1/(um)}\right| \le \left|\frac{e^{itu}-1-itu}{u^2}\right| \le \frac{t^2}{2}.\]
Since $\frac{t^2}{2}$ is integrable over the finite interval $[0,R]$ and $e^{-1/(um)} \to 1$ pointwise for $u>0$, dominated convergence allows us to pass the limit inside the integral. 
Therefore, replacing the dummy variable $u$ with $x$,
\[\lim_{m\to\infty}\ln\Psi_{Z_m}(t)=it(\ln R-\gamma)+\int_0^R\frac{e^{itx}-1-itx}{x^2}\,dx.\]
It remains to evaluate the integral
\[I(R)=\int_0^R\frac{e^{itx}-1-itx}{x^2}\,dx.\]
Applying integration by parts with $u=e^{itx}-1-itx$ and $dv=\frac{dx}{x^2}$, so that $du=it(e^{itx}-1)\,dx$ and $v=-\frac{1}{x}$, we have 
\[I(R)=-\frac{e^{itR}-1-itR}{R}+it\int_0^R\frac{e^{itx}-1}{x}\,dx.\]
Using Euler's formula $e^{itx}=\cos(tx)+i\sin(tx)$, the remaining integral becomes
\begin{align*}
it\int_0^R\frac{e^{itx}-1}{x}\,dx&=it\int_0^R\frac{\cos(tx)-1}{x}\,dx-t\int_0^R\frac{\sin(tx)}{x}\,dx.
\end{align*}
Substituting $y=tx$ gives
\begin{align*}
it\int_0^R\frac{e^{itx}-1}{x}\,dx&=-it\int_0^{tR}\frac{1-\cos y}{y}\,dy-t\int_0^{tR}\frac{\sin y}{y}\,dy \\&=-it\Cin(tR)-t\Si(tR).
\end{align*}
Combining the previous identities, we have
\[C_2(t)=it(\ln R-\gamma)-\frac{e^{itR}-1-itR}{R}-it\Cin(tR)-t\Si(tR),\]
which gives the desired claim after exponentiating. 
\end{proof}

\section{The Algorithmic Oracle: Sampling from \texorpdfstring{$\mathcal{Z}$}{Z}}
\label{sec:oracle}

\Cref{alg:oracle:invert} and \Cref{alg:oracle} generate a sample $z^* \sim \mathcal{Z}$ using inverse transform sampling applied to the cumulative distribution function $F_{\mathcal{Z}}(x)$. 
We use the analytic form of the Gil-Pelaez inversion formula \cite{gilpelaez1951note} to evaluate this CDF: 
\begin{align*}
F_{\mathcal{Z}}(x) &= \frac{1}{2} - \frac{1}{\pi} \int_0^\infty \frac{\Im\left(e^{-itx}\Psi(t)\right)}{t}\,dt \\
&= \frac{1}{2} + \frac{1}{\pi} \int_0^\infty \frac{e^{itx}\Psi(-t)-e^{-itx}\Psi(t)}{2it}\,dt.
\end{align*}
We then use the fact that $\Psi(-t)=\overline{\Psi(t)}$ in the oracle implementation in \Cref{alg:oracle}. 

\begin{algorithm*}[!htb]
\caption{Gil-Pelaez oracle for $F_{\calZ}$}
\label{alg:oracle}
\begin{algorithmic}[1]
\Require{$x\in\mathbb{R}$, truncation parameter $R$ from \Cref{alg:l2:generator}}
\Ensure{Approximation to $F_{\calZ}(x)$}
\State{$U_{\max}=\O{\log n}$, $N=\O{\log^2 n}$, $S=0$}
\State{Retrieve Gauss-Legendre nodes $\{w_m,v_m\}_{m\in[N]}$}
\For{each $m\in[N]$}
\State{$t_m=\frac{U_{\max}}{2}(w_m+1)$}
\State{$\phi_m=\exp\left(it_m(\ln R-\gamma)-\frac{e^{it_mR}-1-it_mR}{R}-it_m\Cin(t_mR)-t_m\Si(t_mR)\right)$}
\State{$I_m=\frac{1}{2it_m}\left(e^{it_mx}\overline{\phi_{m}}-e^{-it_mx}\phi_m\right)$}
\State{$S\gets S+v_m\cdot\frac{U_{\max}}{2}\cdot I_m$}
\EndFor
\State{\Return $\frac{1}{2}+\frac{1}{\pi}S$}
\end{algorithmic}
\end{algorithm*}

\subsection{Domain Truncation via Exponential Decay}
\label{sec:domain:truncation}
A direct evaluation of the Gil-Pelaez integral over an unbounded domain is not computationally feasible. 
The algorithm therefore relies on structural decay of the characteristic function $\Psi(t)$ to truncate the integration range while controlling error.

If $\Psi(t)$ decayed only polynomially, any truncation achieving error $\nu$ would typically require a polynomial integration bound $U_{\max}$, leading to a corresponding polynomial number of quadrature nodes. 
In contrast, the exponential decay established below permits truncation at $U_{\max} = \Theta(\log n)$, which is sufficient for the desired accuracy and preserves polylogarithmic complexity.

\begin{restatable}[Exponential Decay and Domain Truncation]{lemma}{lemtruncation}
\label{lem:truncation}
For $t \in \mathbb{R}$, the characteristic function satisfies the global bound
\[|\Psi(t)| \le \exp\left(\frac{2}{R} - \kappa |t|\right),\]
where $\kappa = \inf_{x \ge 1} \Si(x) > 0$. 
Consequently, truncating the Gil-Pelaez integral to $U_{\max} = \Theta(\log n)$ bounds the truncation error by at most $\nu/2$.
\end{restatable}
\begin{proof}
Recall from \Cref{thm:limit_cf} that $\Psi(t) = \exp(C_2(t))$, where $C_2(t) = it(\ln R-\gamma) - \frac{e^{itR}-1-itR}{R} - it\Cin(tR) - t\Si(tR)$ and $\gamma$ denotes the Euler-Mascheroni constant. 
Also recall that for any complex number $z = A + iB$, $|e^z| = |e^A e^{iB}| = e^A$, since $|e^{iB}| = 1$. 
Thus, $|\Psi(t)| = \exp\left(\Re(C_2(t))\right)$, so it suffices to bound the real part of the exponent.

For real $t \in \mathbb{R}$, the terms $it(\ln R - \gamma)$ and $-it\Cin(tR)$ are purely imaginary and therefore contribute nothing to $\Re(C_2(t))$. 
Expanding the remaining term using Euler's formula gives
\begin{align*}
\Re(C_2(t))&=\Re\left(-\frac{e^{itR} - 1 - itR}{R}\right)- t\Si(tR) \\
&=\Re\left(-\frac{\cos(tR) + i\sin(tR) - 1 - itR}{R}\right)- t\Si(tR) \\
&=\frac{1-\cos(tR)}{R} - t\Si(tR).
\end{align*}
Moreover, recall that $\Si(x)=\int_0^x \frac{\sin u}{u}\,du$. 
Since $\frac{\sin u}{u}$ is even, $\Si$ is odd, so $t\Si(tR)=|t|\Si(|t|R)$. 
We split the analysis into two regimes based on $|t|R$. 

First, consider the regime $|t|\ge\frac{1}{R}$, i.e., $|t|R\ge1$. 
Since $\cos(tR)\in[-1,1]$, we have 
\[\frac{1-\cos(tR)}{R}\le\frac{2}{R}.\] 
Over the interval $[1,\infty)$, the function $\Si(x)$ remains strictly positive, so
\[-t\Si(tR)\le -\kappa |t|,\]
where $\kappa := \inf_{x\ge1}\Si(x) > 0$. 
Combining the bounds, for the regime $|t|R \ge 1$ we have
\[\Re(C_2(t))\le\frac{2}{R}-\kappa |t|.\]

Alternatively, consider the complementary regime $|t|R < 1$. 
Observe that the function $f(u) = 1-\cos u - u\Si(u)$ satisfies $f(0)=0$ and $f'(u) = \sin u - \left(\Si(u) + u\frac{\sin u}{u}\right) = -\Si(u)$. 
Since $\Si(u) \ge 0$ for $u \ge 0$, we have $f'(u) \le 0$, which implies $f(u) \le 0$ for all $u \ge 0$. 
Setting $u = |t|R$, we obtain $\Re(C_2(t)) = \frac{f(|t|R)}{R} \le 0$. 
Furthermore, because $\kappa = \inf_{x\ge 1}\Si(x) \le \Si(1) = \int_0^1 \frac{\sin u}{u}\,du < \int_0^1 1\,du = 1$, we have $\kappa < 2$. 
Thus, for $|t| < \frac{1}{R}$, we have $\kappa|t| < \frac{\kappa}{R} < \frac{2}{R}$, which implies $\frac{2}{R} - \kappa|t| > 0$. 
Therefore, in this regime we also trivially have
\[\Re(C_2(t)) \le 0 < \frac{2}{R} - \kappa |t|.\]
Hence, the bound 
\[|\Psi(t)|\le\exp\left(\frac{2}{R}-\kappa |t|\right)\]
for all $t \in \mathbb{R}$.

We now bound the truncation error in the Gil-Pelaez inversion formula due to the integrand $\frac{1}{\pi} \int_0^\infty \frac{e^{itx}\Psi(-t)-e^{-itx}\Psi(t)}{2it}\,dt$. 
Using the triangle inequality, $|e^{\pm itx}|=1$ for real $t, x$, and the Hermitian symmetry $|\Psi(-t)|=|\Psi(t)|$, we obtain
\[\left|\frac{e^{itx}\Psi(-t)-e^{-itx}\Psi(t)}{2it}\right|\le\frac{|\Psi(t)|}{t}.\]
For sufficiently large $U_{\max}\ge 1$, we additionally have $\frac{1}{t}\le 1$ over the integration domain, so
\begin{align*}
\int_{U_{\max}}^\infty \frac{|\Psi(t)|}{t}\,dt&\le\int_{U_{\max}}^\infty\frac{\exp\left(\frac{2}{R}-\kappa |t|\right)}{t}\,dt\\
&\le\int_{U_{\max}}^\infty e^{2/R-\kappa t}\,dt =\frac{1}{\kappa}e^{2/R-\kappa\cdot U_{\max}}.
\end{align*}
By \Cref{lem:bounds}, we have $\frac{1}{R}=\Theta(\log n)$, so it suffices to set $U_{\max}=\Theta(\log n)$ to achieve a desired truncation bound of $\nu=n^{-c}$ while preserving polylogarithmic complexity.
\end{proof}

\subsection{Gauss-Legendre Quadrature on a Bounded Domain}
\label{sec:quadrature}
We approximate the truncated integral
\[\int_0^{U_{\max}} g(t)\,dt, \quad g(t) = \frac{e^{itx}\Psi(-t)-e^{-itx}\Psi(t)}{2it},\]
using Gauss-Legendre quadrature.

\paragraph{Analyticity considerations.}
The exponential decay from \Cref{lem:truncation} ensures that the integrand remains well-controlled on a finite interval. This avoids the need to evaluate the characteristic function in regimes where the oscillatory term $e^{itR}$ becomes numerically unstable.
Crucially, by utilizing the algebraic continuation $g(t)$, the integrand is an entire function, entirely avoiding the Cauchy-Riemann violations associated with the imaginary projection. 
Restricting to $U_{\max} = \Theta(\log n)$ ensures that $|tR| = \Theta(1)$ uniformly over the complex integration domain, which keeps all special-function evaluations bounded and compatible with standard quadrature convergence theory \cite{trefethen2019approximation}.

\begin{restatable}[Polylogarithmic CDF Oracle]{theorem}{thmoracle}
\label{thm:oracle}
Given $\nu = n^{-c}$, the CDF $F_{\mathcal{Z}}(x)$ can be evaluated to absolute additive error $\nu$ using $N = \mathcal{O}(\log^2 n)$ Gauss-Legendre nodes.
\end{restatable}
\begin{proof}
Let $g(t) = \frac{e^{itx}\Psi(-t)-e^{-itx}\Psi(t)}{2it}$, which is the analytic continuation used to compute the Gil-Pelaez integrand over the truncated interval $[0,U_{\max}]$. 
Because $\Psi(t)$ is composed of entire functions and the quotient has a removable singularity at $t=0$, $g(t)$ extends analytically to the entire complex plane. 
Thus, classical results for Gauss-Legendre quadrature apply \cite{trefethen2019approximation}. 
The convergence rate is determined by how far the integrand can be analytically continued into the complex plane around the integration interval. 
This neighborhood is measured using a Bernstein ellipse $E_\rho$, i.e., an ellipse in the complex plane with foci at $\pm 1$ that encloses $[-1,1]$, whose size controls the exponential convergence rate; its \emph{parameter} $\rho > 1$ is a single scalar describing how ``large'' the ellipse is (larger $\rho$ corresponds to a fatter ellipse and faster convergence). 
If $g$ remains analytic inside $E_\rho$, then the error for the quadrature evaluation $Q_N(g)$ satisfies
\[\left|\int_0^{U_{\max}} g(t)\,dt - Q_N(g)\right|\le\O{U_{\max}\frac{M_g}{\rho^{2N}}},\]
where $M_g = \sup_{z \in E_\rho} |g(z)|$~\cite{trefethen2019approximation}, and the $U_{\max}$ factor accounts for the affine scaling of the integration interval.
We map the standard interval $[-1,1]$ to $[0,U_{\max}]$ via $t = \frac{U_{\max}}{2}(w+1)$. 
To keep the analytic continuation away from large imaginary values, we choose the ellipse in the $t$-plane so that its imaginary semi-axis remains bounded by a constant $|\Im(t)| \le 1$.
Under the affine map above, this corresponds to a semi-minor axis $b_w = \frac{2}{U_{\max}}$ in the $w$-plane. 
The corresponding Bernstein parameter is
\[\rho = \sqrt{1+b_w^2}+b_w = \sqrt{1+\left(\frac{2}{U_{\max}}\right)^2}+\frac{2}{U_{\max}}.\]
Since $U_{\max} = \Theta(\log n)$ by \Cref{lem:truncation}, $\rho = 1 + \Theta\left(\frac{1}{U_{\max}}\right)$, and therefore $\ln \rho = \Theta\left(\frac{1}{\log n}\right)$. 

It remains to bound the maximum modulus $M_g$. 
Every point on the contour can be written as $t = u+iv$, where $u \in [0,U_{\max}]$ and $|v| \le 1$. 
For the phase factors, we have
\[|e^{\pm itx}| = |e^{\pm i(u+iv)x}| = e^{\mp vx} \le e^{|x|}.\]
Since $x = \O{\log n}$ over the bounds of the binary search interval, this complex phase factor is at most $e^{\O{\log n}}$, contributing at most $\O{\log n}$ upon taking the logarithm.
Near the removable singularity at the origin $t=0$, the analytic continuation evaluates via L'H\^opital's rule to $\lim_{t\to 0} g(t) = x + i\Psi'(0) = x - (\ln R - \gamma) = \O{\log n}$. 
Thus $g(t)$ remains analytically bounded near the origin. 
For $t$ bounded away from the origin, the dominant contribution to the modulus comes from the characteristic functions evaluated at complex arguments $\pm t$:
\[|\Psi(\pm t)| = \exp\left(\Re(C_2(\pm t))\right),\]
where we recall $C_2(t) = it(\ln R-\gamma) - \frac{e^{itR}-1-itR}{R} - it\Cin(tR) - t\Si(tR)$ by \Cref{thm:limit_cf}. 
We now upper bound $\Re(C_2(\pm t))$ term by term. 

First, consider the drift term
\[\pm it\ln R = \pm i(u+iv)\ln R = \mp v\ln R \pm iu\ln R.\]
The imaginary component contributes only to phase, while the real component satisfies $|\mp v\ln R|\le \ln(1/R)$.
By \Cref{lem:bounds}, $\frac{1}{R} = \Theta(\log n)$, so this term contributes at most $\O{\log\log n}$. 

Next, consider the exponential remainder term. 
Since $|tR| \le (U_{\max}+1)R$, $U_{\max} = \Theta(\log n)$, and $R = \Theta\left(\frac{1}{\log n}\right)$, we obtain $|tR| = \Theta(1)$. 
Consequently,
\[|e^{\pm itR}| = |e^{\pm iuR}e^{\mp vR}| = e^{\mp vR} \le e^R = \mathcal{O}(1).\]
Thus the term $-\frac{e^{\pm itR}-1 \mp itR}{R}$ has real part upper bounded by $\O{\frac{1}{R}}=\O{\log n}$.

Finally, over the bounded region $|tR|=\Theta(1)$, the entire functions $\Si(\pm tR)$ and $\Cin(\pm tR)$ remain uniformly bounded. 
Since $|t|\le U_{\max}+1 = \Theta(\log n)$, the terms $\mp it\Cin(\pm tR)$ and $\mp t\Si(\pm tR)$ contribute at most $\O{\log n}$ to the real part.

Combining these estimates, we have
\[\ln M_g = \sup_{t\in E_\rho}\Re(C_2(t)) = \O{\log n}.\]
To achieve quadrature error at most $n^{-c}$, it suffices that $U_{\max}\frac{M_g}{\rho^{2N}}\le n^{-c}$. 
Since $\log\rho = \Theta\left(\frac{1}{\log n}\right)$, $\log U_{\max} = \O{\log\log n}$, and $\log M_g = \O{\log n}$, it suffices to set $N=\O{\log^2 n}$ quadrature node evaluations to achieve additive error $n^{-c}$. 
\end{proof}

\section{Statistical Recovery and End-to-End Correctness}
\label{sec:end-to-end}
To finalize the perfect $L_2$ sampling framework, we formalize the protocol to extract the heavy-hitter index from the compressed linear sketch and prove that the returned index matches the target distribution.

\subsection{The Dense Post-Processing Statistical Test}
\label{sec:statistical:test}
Let $z \in \mathbb{R}^{N}$ represent the fully duplicated frequency vector. 
Recall that in a dense $\CountSketch$, the matrix $G \in \mathbb{R}^{R\times N}$ consists of independent standard Gaussian entries, c.f., \Cref{alg:dense:countsketch}. 
To recover an estimate for any explicitly tracked virtual item $z_{i,j}$, the algorithm computes the scaled dot product between the sketch vector $y = Gz$ and the explicit tracking weights $G_{i,j}$. 
To handle the variance, we partition the $R$ rows into $M = \Theta(\log n)$ independent blocks, each containing $r = \Theta(\log n)$ rows. 
Let $Z$ be a $2$-approximation of $\|z\|_2$.

\begin{figure*}[!htb]
\begin{mdframed}
\begin{enumerate}
\item
Construct the vector $z\in\mathbb{R}^{n'}$ with entries given by $\{\calH_i,N_i\}$, where $n'=\O{n\log n}$.
\item 
Let $\widehat{z_{D(1)}}$ and $\widehat{z_{D(2)}}$ denote the dense $\CountSketch$ estimates of the two largest coordinates in $\cup_{i=1}^n\calH_i$, using $R$ rows.
\item
Let $Z$ be a $2$-approximation to $\|z\|_2$ computed from the dense $\CountSketch$.
\item
Sample $\mu$ uniformly from the interval $[0.99,1.01]$, where $\eps$ is a fixed parameter.
\item
If $|\widehat{z_{D(1)}}|-|\widehat{z_{D(2)}}|<100\mu\eps Z$ or $|\widehat{z_{D(2)}}|<50\mu\eps Z$, return $\FAIL$.
\item
Otherwise, the statistical test succeeds.
\end{enumerate}
\end{mdframed}
\caption{Statistical test for the perfect $L_2$ sampler.}
\label{fig:stat:test}
\end{figure*}

We directly leverage the statistical gap test properties proven in \cite{SwartworthWZ25}. 
The following statements assert that the dense $\CountSketch$ gap test does not bias the output distribution, and that it succeeds with constant probability.

\begin{lemma}[Lemma 4.2 in \cite{SwartworthWZ25}]
\label{lem:swz_4_2}
For any constant $p\in(0,2]$ and $\nu =\Theta(n^{-\zeta/20})$, we have $\PPr{\neg \FAIL \mid D(1)} = \PPr{\neg\FAIL} \pm n^{-\O{c}}$ for all possible $D(1) \in [N]$.
\end{lemma}

\begin{lemma}[Lemma 4.7 in \cite{SwartworthWZ25}]
\label{lem:swz_4_7}
Let $\eps = \frac{1}{\Theta(\sqrt{\log n})}$. 
Then with probability at least $0.75$, we have both $|z_{D(1)}| - |z_{D(2)}| = \Theta(\sqrt{k} \|x\|_2)$ and the statistical test succeeds.
\end{lemma}

Putting together \Cref{lem:swz_4_2} and \Cref{lem:swz_4_7}, we have:

\begin{restatable}[Ideal Distribution and Target Gap]{lemma}{lemidealdist}
\label{lem:ideal_dist}
Let $D(1)$ be the coordinate that maximizes the unsketched, perfect virtual vector $z$. 
Then we have $\PPr{D(1) = j} = \frac{x_j^2}{\|x\|_2^2}$. 
\end{restatable}
\begin{proof}
To find the maximum magnitude in the virtual vector $z$, we can equivalently look for the minimum of the squared reciprocal values. 
Let 
\[M_{i,j} := \frac{1}{z_{i,j}^2} = \frac{E_{i,j}}{x_i^2}.\]
Because $E_{i,j} \sim \Exp(1)$, standard scaling of exponential random variables implies that $M_{i,j}$ follows an exponential distribution with rate parameter $\lambda_i = x_i^2$, denoted $M_{i,j} \sim \Exp(x_i^2)$. 
Moreover, the global maximum magnitude $|z_{D(1)}|$ corresponds to the global minimum $V_1 := \min_{i,j} M_{i,j}$. 
Since the collection $\{M_{i,j}\}$ consists of $n \times k$ mutually independent exponential random variables, the probability that the minimum belongs to a specific copy $j$ of coordinate $c$ is given by the ratio of its rate to the total sum of all rates:
\[\PPr{M_{c,j} = V_1} = \frac{x_c^2}{\sum_{i=1}^n \sum_{m=1}^k x_i^2} = \frac{x_c^2}{k \|x\|_2^2}.\]
Because each coordinate $c$ is duplicated $k$ times, the total probability that the global maximum corresponds to coordinate $c$ is the sum of these mutually exclusive events:
\[\PPr{D(1) \in C_c} = \sum_{j=1}^k \frac{x_c^2}{k \|x\|_2^2} = \frac{x_c^2}{\|x\|_2^2}.\]
Hence, we recover the desired $L_2$ sampling distribution. 
\end{proof}

\subsection{Bounding the Algorithmic Simulation Error}
\label{sec:simulation:error}
To show that the sampler is statistically close to the desired distribution, we upper bound two sources of approximation error:
\begin{enumerate}
\item 
\textbf{Head approximation error:} 
The discrepancy between the finite distribution governing the largest coordinates and its continuous stable limit.
\item 
\textbf{Tail approximation error:} 
The perturbation introduced by numerically evaluating the limiting tail distribution, including the resulting Kolmogorov-Smirnov (KS) error from the quadrature procedure.
\end{enumerate}

Instead of directly coupling the characteristic functions of the finite and limiting constructions, we work in the point process framework developed by \cite{SwartworthWZ25}, which is more stable for analyzing extreme order statistics.

\begin{lemma}[Lemma 4.14 in \cite{SwartworthWZ25}]
\label{lem:head_tv}
Let $A = [y_0, \infty)$ denote the region capturing the top $\tau = \Theta(\log n)$ normalized inverse exponential variables. 
Then the total variation distance between the finite point process induced by the exact construction and its limiting continuous point process, restricted to $A$, is bounded by $\O{n^{-C/2}}$.
\end{lemma}

The tail component is obtained by numerically evaluating a truncated limiting $1$-stable CDF using Gauss-Legendre quadrature. 
This introduces a controlled additive error in the evaluations of the limiting CDF. 
Rather than translating this error into a point-wise perturbation of the sampled quantile or of the induced variance, we compare the resulting numerical sampler directly with the limiting distribution in total variation distance. 
We then propagate this guarantee through the variance transformation and the Gaussian sketch using the data processing inequality.

\begin{restatable}[Numerical Tail Sampling in Total Variation]{lemma}{lemshift}
\label{lem:shift}
Fix constants $0<a<b$. Suppose $R\in\left[\frac{a}{\log n},\frac{b}{\log n}\right]$. 
Let $\calZ$ denote the limiting distribution from \Cref{thm:limit_cf}, with CDF $F_{\calZ}$ and density $f_{\calZ}$. 
Suppose \Cref{alg:oracle:invert} performs binary search on $[-B,B]$, terminates when the interval width is at most $h$, and, after termination, returns a uniformly random point in the final interval. 
Suppose furthermore that every CDF evaluation has additive error at most $\eta$. 
Let $\widehat Z$ denote the resulting sample. 
Then
\[d_{TV}(\widehat Z,\calZ)\le\O{\frac{B}{h}\eta}+\PPr{|\calZ|>B}+\O{Bh\|f'_{\calZ}\|_\infty}.\]
Consequently, for every fixed $c>0$, by taking $B=\Theta(\log n)$, $h=n^{-C_{\mathrm{bin}}}$, and $\eta=n^{-C_{\mathrm{cdf}}}$ for sufficiently large constants $C_{\mathrm{bin}},C_{\mathrm{cdf}}$, we have $d_{TV}(\widehat Z,\calZ)\le \O{n^{-c}}$. 
\end{restatable}
\begin{proof}
Let $\overline Z$ denote the output of the same binary-search procedure when every comparison uses the exact CDF $F_{\calZ}$, coupled to $\widehat Z$ using the same random variables $U$ and $W$.

The binary-search tree has $\O{B/h}$ possible internal comparison points. 
At any such point $q$, the exact and numerical comparisons can differ only when $U$ lies between $F_{\calZ}(q)$ and its numerical approximation. 
Since the CDF approximation error is at most $\eta$, this occurs with probability at most $\eta$. 
A union bound therefore gives
\[d_{TV}(\widehat Z,\overline Z)\le \O{\frac{B}{h}\eta}.\]
We next compare $\overline Z$ with $\calZ$. 
On every interior terminal interval $I$ of width at most $h$, the density of $\overline Z$ is the average of $f_{\calZ}$ over $I$. 
Hence, for $x\in I$,
\[\left|f_{\calZ}(x)-\frac{1}{|I|}\int_I f_{\calZ}(y)\,dy\right|\le h\|f'_{\calZ}\|_\infty.\]
Summing over the intervals contained in $[-B,B]$ gives
\[d_{TV}(\overline Z,\calZ)\le\PPr{|\calZ|>B}+\O{Bh\|f'_{\calZ}\|_\infty}.\]
Thus
\[d_{TV}(\widehat Z,\calZ)\le\O{\frac{B}{h}\eta}+\PPr{|\calZ|>B}+\O{Bh\|f'_{\calZ}\|_\infty}.\]
It remains to bound the last two terms. 
From the proof of \Cref{thm:limit_cf}, for $s\in\{-1,1\}$,
\[\log \Ex{\exp(s\calZ)}=s(\log R-\gamma)+\int_0^R\frac{e^{su}-1-su}{u^2}\,du.\]
Since $R=\Theta(1/\log n)$, both exponential moments are $n^{o(1)}$. 
Hence, by Markov's inequality, for $B=C_B\log n$ with $C_B$ sufficiently large,
\[\PPr{|\calZ|>B}\le n^{-c-2}.\]

Moreover, through Fourier inversion and \Cref{lem:truncation}, we have
\[\|f'_{\calZ}\|_\infty\le\frac{1}{2\pi}\int_{\mathbb R}|t||\Psi_{\calZ}(t)|\,dt\le\frac{1}{2\pi}\int_{\mathbb R}|t|\exp\left(\frac{2}{R}-\kappa|t|\right)\,dt=n^{\O{1}}.\]
Taking $h=n^{-C_{\mathrm{bin}}}$ and $\eta=n^{-C_{\mathrm{cdf}}}$ for sufficiently large constants finishes the proof. 
The required CDF accuracy is achievable with $\O{\log^2 n}$ quadrature nodes by \Cref{thm:oracle}.
\end{proof}

\begin{restatable}[Tail Simulation under Post-Processing]{lemma}{lemperturbationrobust}
\label{lem:perturbation_robustness}
Let $Z$ and $\widetilde Z$ be random variables satisfying $d_{TV}(Z,\widetilde Z)\le\eps$. 
Let $T(z)=\max(0,m(z+\log m))$. 
Then $d_{TV}(T(Z),T(\widetilde Z))\le\eps$. 
Furthermore, for any fixed $x\in\mathbb R$ and any number $s$ of sketch rows, if $G\sim\mathcal N(0,I_s)$ is independent randomness, then
\[d_{TV}\left(x\sqrt{T(Z)}\,G,\,x\sqrt{T(\widetilde Z)}\,G\right)\le\eps.\]
The same bound holds after any subsequent deterministic or randomized post-processing of these sketch values.
\end{restatable}
\begin{proof}
The map $z\mapsto T(z)$ is deterministic, so the first claim follows from the data processing inequality for total variation distance. 
Conditioned on $T(Z)=v$, the Gaussian sketch contribution is generated from the Markov kernel 
\[v\mapsto\mathcal N(0,x^2v I_s).\]
Applying a Markov kernel cannot increase total variation distance, which proves the second claim. 
The final statement follows from another application of data processing. 
\end{proof}

\begin{restatable}[Rate of convergence for tail simulation]{lemma}{lemgaussiantv}
\label{lem:gaussian_tv}
Let $Z_m=\frac{1}{m}\sum_{j=1}^m X_j-\log m$, where the variables $X_j$ are independent with probability density function $f_X(x)=C_R x^{-2}e^{-1/x}\mathbf{1}_{(0,Rm]}(x)$, where $C_R=\frac{1}{\PPr{Y\le Rm}}
\le 2$. 
Suppose $m=\Omega(n^C)$ and for fixed constants $0<a<b<$, $R\in\left[\frac{a}{\log n},\frac{b}{\log n}\right]$. 
Then for every fixed $c>0$ and for sufficiently large $C>0$, we have
\[d_{TV}(Z_m,\calZ)\le \O{n^{-c}},\]
where $\calZ$ denotes the limiting distribution.
\end{restatable}
\begin{proof}
Let $\Psi_{Z_m}$ and $\Psi_{\calZ}$ denote the characteristic functions of $Z_m$ and $\calZ$, and let $f_{Z_m}$ and $f_{\calZ}$ denote their probability density functions. 
By Fourier inversion,
\[\|f_{Z_m}-f_{\calZ}\|_\infty\le\frac{1}{2\pi}\int_{-\infty}^{\infty}|\Psi_{Z_m}(t)-\Psi_{\calZ}(t)|\,dt.\]
We split the integral into low and high frequency regimes.

For $|t|\le m^{1/4}$, the calculation in the proof of \Cref{thm:limit_cf} can be made quantitative. 
Uniformly for $R\in[a/\log n,b/\log n]$,
\[\left|\log\Psi_{Z_m}(t)-\log\Psi_{\calZ}(t)\right|\le\O{\frac{t^2\log^2 m+|t|\log m\log n}{m}}.\]
Indeed, writing
\[y_m(t)=\EEx{}{\exp(itX/m)}-1,\]
we have $|y_m(t)|=\O{|t|\log m/m}$, so the error in replacing $m\log(1+y_m(t))$ by $my_m(t)$ is $\O{t^2\log^2m/m}$. 
The remaining difference follows from $1-e^{-1/(um)}\le\min\{1,(um)^{-1}\}$ together with
\[E_1(1/(Rm)) = -\gamma+\log R+\log m+\O{1/(Rm)}.\]
Since $|t|\le m^{1/4}$, the displayed error is $o(1)$ uniformly.
Consequently,
\[|\Psi_{Z_m}(t)-\Psi_{\calZ}(t)|\le|\Psi_{\calZ}(t)|\O{\frac{t^2\log^2 m+|t|\log m\log n}{m}}.\]
Using \Cref{lem:truncation} and integrating gives
\[\int_{|t|\le m^{1/4}}|\Psi_{Z_m}(t)-\Psi_{\calZ}(t)|\,dt\le \frac{n^{\O{1}}}{m}.\]

We now consider the region $|t|>m^{1/4}$. 
For $\Psi_{\calZ}$, the exponential decay immediately gives
\[\int_{|t|>m^{1/4}}|\Psi_{\calZ}(t)|\,dt\le\exp\left(-\Omega(m^{1/4})\right).\]
It remains to bound $\Psi_{Z_m}$. 
Note that $\Psi_{Z_m}(t)=\Psi_X\left(\frac{t}{m}\right)^m$ and let $\xi=\frac{t}{m}$. 
We split the analysis into two further regimes.

First consider $m^{-3/4}<|\xi|\le \eps$, where $\eps\in(0,1/16]$ is a sufficiently small constant. 
Since the variance of $X$ diverges with $Rm$, a quadratic Taylor expansion is not useful here. 
Instead, we use the identity
\[1-|\Psi_X(\xi)|^2 = \Ex{1-\cos\left(\xi(X_1-X_2)\right)},\]
where $X_1,X_2$ are independent copies of $X$.
Set $A=\frac{1}{4|\xi|}$. 
Restricting to the event $X_1\in[1,2]$ and $X_2\in[A,2A]$, we have $|\xi(X_2-X_1)|\in[1/8,1/2]$. 
Hence
\[1-\cos\left(\xi(X_2-X_1)\right)\ge1-\cos(1/8)=: \kappa'>0.\]
We next lower bound the probability of these intervals. 
Since $f_X(x)\ge x^{-2}e^{-1/x}$ on $(0,Rm]$, then $\PPr{X_1\in[1,2]}\ge\frac{1}{2e}$ and
\begin{align*}
\PPr{X_2\in[A,2A]}\ge\int_A^{2A}x^{-2}e^{-1/x}
\,dx \ge e^{-1/A}\left(\frac{1}{A}-\frac{1}{2A}\right)\ge\frac{2}{e^{1/4}}|\xi|.
\end{align*}
Combining the two estimates gives $1-|\Psi_X(\xi)|^2\ge c_0|\xi|$ for some constant $c_0>0$. 
Using $\sqrt{1-y}\le 1-\frac{y}{2}$, we obtain
\[|\Psi_X(\xi)|\le1-\frac{c_0}{2}|\xi|.\]
Substituting $\xi=t/m$ therefore gives
\[|\Psi_{Z_m}(t)| = |\Psi_X(t/m)|^m \le \left(1-\frac{c_0|t|}{2m}\right)^m \le \exp\left(-\frac{c_0}{2}|t|\right).\]
Integrating over $m^{1/4}\le |t|\le \eps m$ gives
\[\int_{m^{1/4}\le |t|\le \eps m}|\Psi_{Z_m}(t)| \,dt \le \exp\left(-\Omega(m^{1/4})\right).\]
Finally, consider the regime $|\xi|>\eps$. 
Let $Y$ denote the untruncated random variable with density $f_Y(y)=y^{-2}e^{-1/y}$. 
Since $Y$ has an absolutely continuous distribution, the Riemann-Lebesgue lemma implies $\Psi_Y(\xi)\to 0$ as $|\xi|\to\infty$. 
By compactness, there exists $\delta>0$ such that $|\Psi_Y(\xi)|\le1-2\delta$ for all $|\xi|\ge\eps$. 
The truncated characteristic function differs from $\Psi_Y$ by at most the discarded tail mass:
\[|\Psi_X(\xi)-\Psi_Y(\xi)|\le\frac{\PPr{Y>Rm}}{\PPr{Y\le Rm}}\le\O{\frac{1}{Rm}}.\]
Hence for sufficiently large $m$, we have $|\Psi_X(\xi)|\le 1-\delta$ for all $|\xi|>\eps$. 

We now estimate the remaining integral:
\begin{align*}
\int_{|t|>\eps m}|\Psi_{Z_m}(t)|\,dt&=m\int_{|\xi|>\eps}|\Psi_X(\xi)|^m\,d\xi \\
&\le m(1-\delta)^{m-2}\int_{-\infty}^{\infty}|\Psi_X(\xi)|^2\,d\xi.
\end{align*}
By Parseval's identity,
\[\int_{-\infty}^{\infty}|\Psi_X(\xi)|^2\,d\xi=2\pi\int_0^{Rm}f_X(x)^2\,dx.\]
Since $f_X(x)^2\le4x^{-4}e^{-2/x}$, the integral is finite and bounded by an absolute constant. 
Therefore,
\[\int_{|t|>\eps m}|\Psi_{Z_m}(t)|\,dt\le\exp\left(-\Omega(m)\right).\]
Combining the low and high frequency estimates gives
\[\|f_{Z_m}-f_{\calZ}\|_\infty\le\frac{\poly(n)}{m}.\]
Finally, the same cumulant formulas evaluated at $t=\pm i$ show that
\[\Ex{e^{Z_m}}+\EEx{}{e^{-Z_m}}+\Ex{e^{\calZ}}+\EEx{}{e^{-\calZ}}\le n^{\O{1}}.\]
Thus, for $B=C_B\log n$ with $C_B$ sufficiently large,
\[\PPr{|Z_m|>B}+\PPr{|\calZ|>B}\le n^{-c-2}.\]
Therefore,
\[d_{TV}(Z_m,\calZ)\le B\|f_{Z_m}-f_{\calZ}\|_\infty+n^{-c-2}\le\frac{n^{\O{1}}}{m}+n^{-c-2}.\]
Since $m=\Theta(n^C)$, choosing $C$ sufficiently large proves the claim.
\end{proof}

\begin{restatable}{lemma}{lemtvderror}
\label{lem:tvd:error}
Let the target additive failure rate of the perfect $L_2$ sampler be $n^{-c}$. The numerical execution differs from the exact theoretical distribution by a total variation distance of at most $\O{n^{-c}}$.
\end{restatable}
\begin{proof}
Fix the desired exponent $c>0$, and apply
\Cref{lem:shift,lem:gaussian_tv} with target error $n^{-c-3}$.

Condition first on the explicit head randomness. 
By \Cref{lem:bounds}, choosing the constant in $\tau=\Theta(\log n)$ sufficiently large and taking a union bound over the $n$ original coordinates, with probability at least $1-\O{n^{-c-2}}$ every corresponding truncation parameter satisfies 
\[R_i\in\left[\frac{a}{\log n},\frac{b}{\log n}\right]\]
for fixed constants $0<a<b$.

On this event, for every coordinate $i$, let $Z_{m,i}$ denote its exact finite normalized tail statistic and let $\widehat Z_i$ denote the sample generated by the numerical oracle. 
By the triangle inequality and \Cref{lem:shift,lem:gaussian_tv},
\[d_{TV}(Z_{m,i},\widehat Z_i)\le d_{TV}(Z_{m,i},\calZ_i)+d_{TV}(\calZ_i,\widehat Z_i)\le\O{n^{-c-3}}.\]
Conditioned on the head randomness, the tail variables are independent across the original coordinates. 
Hence, the total variation distance between the joint exact and simulated tail vectors is at most
\[\sum_{i=1}^n d_{TV}(Z_{m,i},\widehat Z_i)\le\O{n^{-c-2}}.\]
Equivalently, the two tail vectors admit a coupling under which they agree coordinate-wise except with probability $\O{n^{-c-2}}$. 
By \Cref{lem:perturbation_robustness}, applying the tail-variance map, generating the Gaussian sketch contributions, and running the statistical recovery procedure cannot increase this total variation distance.

Finally, \Cref{lem:head_tv} contributes $\O{n^{-C/2}}$ from the head construction. 
Choosing the duplication constant $C$ sufficiently large, and adding the probability of the bad truncation event above, gives
\[d_{TV}(\calD_{\mathrm{exact}},\calD_{\mathrm{alg}})\le \O{n^{-c}}.\]
\end{proof}

\subsection{Derandomization}
\label{sec:prg}
The preceding analysis assumes that independent exponential and normal random variables can be generated and stored efficiently. 
To implement this algorithm in the streaming model, we now derandomize these components of the linear sketch. 
We begin by recalling the pseudorandom generator of \cite{GopalanKM18}, which fools a class of Fourier transforms. 

\begin{definition}
An \emph{$(m,n)$-Fourier shape} is a function $f:[m]^n\to\mathbb{C}$ of the form
\[f(x_1,\ldots,x_n)=\prod_{j=1}^n f_j(x_j),\]
where each $f_j:[m]\to\mathbb{C}$ acts on a single coordinate.
\end{definition}
\noindent
The following theorem summarizes the guarantees of the Fourier-shape PRG.

\begin{theorem}
\cite{GopalanKM18}
\label{thm:gkm:prg}
Let $f:[m]^n\to\mathbb{C}$ be an $(m,n)$-Fourier shape with
\[f(x_1,\ldots,x_n)=\prod_{j=1}^n f_j(x_j),\]
where each $f_j:[m]\to\mathbb{C}$. 
Then for every $\eps>0$, there exists a seed length $\ell=\O{\log\frac{mn}{\eps}\cdot\left(\log\log\frac{mn}{\eps}\right)^2}$ and a deterministic function $G:\{0,1\}^{\ell}\to[m]^n$ such that
\[\left|\EEx{x\sim[m]^n}{f(x)}-\EEx{y\sim\{0,1\}^{\ell}}{f(G(y))}\right|\le\eps,\]
i.e., all $(m,n)$-Fourier shapes are $\eps$-fooled by $G$. 
\end{theorem}

Furthermore, the PRG of \cite{GopalanKM18} also fools half-space queries, which we now define. 
For each $i\in[\lambda]$, let a half-space query $H_i:\mathbb{R}^n\to\{0,1\}$ on an input $X=(x_1,\ldots,x_n)$ be given by
\[H_i(X)=\mathbf{1}\left[a_1^{(i)}x_1+\cdots+a_n^{(i)}x_n>\theta_i\right],\]
where $a^{(i)}\in\mathbb{Z}^n$ and $\theta_i\in\mathbb{Z}$ for all $i\in[\lambda]$.

\begin{definition}[$\lambda$-half-space tester]
Given an input $X=(x_1,\ldots,x_n)$, a $\lambda$-half-space tester is a function
\[\sigma\left(H_1(X),\ldots,H_\lambda(X)\right)\in\{0,1\},\]
where $\sigma:\{0,1\}^{\lambda}\to\{0,1\}$ and $H_1,\ldots,H_\lambda$ are half-space queries. 
We say that the $\lambda$-half-space tester is $M$-bounded if each coordinate $x_i$ is sampled from an integer distribution with magnitude at most $M$, and every coefficient $a_j^{(i)}$ and threshold $\theta_i$ also has magnitude at most $M$.
\end{definition}

The following theorem gives the corresponding guarantees for fooling half-space testers, c.f., Lemma 7 and Proposition 8 of \cite{JayaramW18}.

\begin{theorem}
\cite{GopalanKM18,JayaramW18}
\label{thm:gkm:prg:half:space}
Let $\calD$ be a distribution on $\{-M,\ldots,M\}$ with magnitude bounded by $M$ that can be sampled using at most $\O{\log M}$ random bits. 
Let $x_1,\ldots,x_n\sim\calD$ and define $X=(x_1,\ldots,x_n)$. 
Then for any $\eps>0$ and $c\ge1$, there exists $\ell=\O{\lambda\left(\log\frac{nM}{\eps}\right)\left(\log\log\frac{nM}{\eps}\right)^2}$ and a deterministic function $F:\{0,1\}^{\ell}\to\{-M,\ldots,M\}^n$ such that for every $\lambda$-half-space tester $\sigma_H$,
\[\left|\EEx{X\sim\calD^n}{\sigma_H(X)=1}-\EEx{y\sim\{0,1\}^{\ell}}{\sigma_H(F(y))=1}\right|\le\eps.\]
Moreover, if $F(y)\in\{-M,\ldots,M\}^n$, then each coordinate of $F(y)$ can be computed using $\O{\ell}$ space and $\polylog(nM)$ time.
\end{theorem}

\begin{theorem}
\label{thm:jw:prg}
\cite{JayaramW18}
Let $\calS$ be a streaming algorithm that, given a stream vector $f\in\{-M,\ldots,M\}^n$ and a fixed matrix $X\in\mathbb{R}^{t\times nk}$, 
whose entries lie in $\{-M,\ldots,M\}$ for some $M=\poly(n)$, produces an output that depends only on the sketches $Af$ and $X\vec{A}$. 
Suppose that the random matrix $A\in\mathbb{R}^{k\times n}$ has i.i.d. entries, each of which can be sampled using $\O{\log n}$ random bits. 
Let $\sigma:\mathbb{R}^k\times\mathbb{R}^t\to\{0,1\}$ be a tester for the success of $\calS$, i.e., $\sigma\left(Af,X\vec{A}\right)=1$
whenever $\calS$ succeeds. 

Then for every fixed constant $c\ge1$, there exists an implementation of $\calS$ using a random matrix $A'$ generated from a seed of length $\O{(k+t)\log n(\log\log n)^2}$ such that
\[\left|\PPr{\sigma\left(Af,X\vec{A}\right)=1}-\PPr{\sigma\left(A'f,X\vec{A'}\right)=1}\right|<n^{-c(k+t)}.\]
Moreover, every entry of $A'$ can be computed in $\polylog(n)$ time using working space proportional to the seed length.
\end{theorem}

\begin{restatable}[Derandomizing the dense Gaussian sketch]{lemma}{lemgaussianderandomization}
\label{lem:gaussian-derandomization}
Let $s$ be the number of rows of the dense Gaussian CountSketch.
Suppose all inputs and Gaussian entries have $\O{\log n}$ bits of precision.  Then, for every fixed constant $c>0$, the Gaussian randomness used by the dense CountSketch and its recovery procedure can be replaced by pseudorandomness generated from a seed of length $O{s\log n(\log\log n)^2}$ such that the probability of every output of the recovery procedure changes by at most $n^{-c}$. 
Moreover, every required Gaussian entry can be generated from the seed in $\polylog(n)$ time.
\end{restatable}
\begin{proof}
This is the Gaussian-sketch derandomization in the proof of Theorem 5.13 in \cite{SwartworthWZ25}, with the number of sketch rows left as a parameter. 
That argument further applies Theorem 7.13 of \cite{JayaramW18} to the dense linear sketch and uses the GKM half-space PRG. 
If the dense sketch has $s$ rows, the quantities denoted $k+t$ in that argument are $\O{s}$, and hence the seed length is $\O{s\log n(\log\log n)^2}$. 
The argument is otherwise independent of $p$. 
In our application $s=\O{\log^2 n}$, giving seed length $\O{\log^3 n(\log\log n)^2}$. 
\end{proof}

\begin{restatable}{theorem}{thmderandomize}
\label{thm:derandomize}
Suppose the input frequency vector $x \in \mathbb{Z}^n$ has integer entries bounded in magnitude by $\poly(n)$. 
Then \Cref{alg:update} can be derandomized to use $\O{\log^3 n(\log\log n)^2}$ bits of space and uses $\polylog(n)$ update time per arriving element in expectation. 
\end{restatable}
\begin{proof}
We separate the randomness into the Gaussian randomness $R_g$ used by the dense CountSketch and the randomness $R_e$ used to generate the exponential head variables and the simulated tail variables.
We first reduce to finite precision. 
As in the analysis of~\cite{SwartworthWZ25,JayaramW18}, we condition on the high-probability event that every continuous random variable used by the algorithm has magnitude at most $\poly(n)$ and retain $C_{\mathrm{disc}}\log n$ bits of precision, for a sufficiently large constant $C_{\mathrm{disc}}$. 
The resulting perturbation of each maintained sketch value is $n^{-\Omega(C_{\mathrm{disc}})}$. 
Hence, the change in each output probability is at most $n^{-c-10}$, and so we can take values form discrete set of $\poly(n)$-bit magnitude, having $\O{\log n}$ bit complexity.

We next derandomize $R_g$. 
The dense CountSketch used here has $s=\O{\log^2 n}$ rows.  
Therefore, by \Cref{lem:gaussian-derandomization}, its Gaussian matrix can be
generated from a seed $y_1$ of length $\O{\log^3 n(\log\log n)^2}$, while changing every output probability by at most $n^{-c-10}$. 
Every entry of this pseudorandom matrix can be recomputed from $y_1$ in $\polylog(n)$ time. 
Fixing $y_1$, let $G'=F_1(y_1)$ denote the resulting pseudorandom sketch matrix.

Now, consider any fixed seed $F_{1}(y_{1})$ and corresponding tester $\mathcal{B}_{i,F_{1}(y_{1})}(R_{e})$. 
Note that this determines the sketch matrix $G^{\prime}=F_{1}(y_{1})$ and so it suffices to consider the residual tester $\mathcal{B}_{i}(R_{e}) := \mathcal{A}_{i}(G^{\prime}, R_{e})$. 
It remains to derandomize the randomness $R_{e}$ used to generate the exponential random variables. 
Because we truncate and round each of the coordinates to $\nu_{disc}$ granularity, it suffices to focus on the discrete exponential random variables on a set with bounded support $\poly(n)$.

Instead of decomposing $R_{e}$ into separate head and tail components, observe that the turnstile stream updates are linear. 
The simulated randomness $r_{j}$ used for any single stream coordinate $j \in [n]$ operates completely independently of the randomness used for other coordinates. 
For each coordinate $j$, the algorithm uses the local seed $r_{j}$ to generate both its head order statistics and its simulated tail variance, which deterministically constructs a vector $v_{j}(r_{j}) \in \mathbb{R}^{n^{\prime}}$ representing the sketch updates for that coordinate.

The final virtual vector $z(R_{e})$ is exactly the sum of these independent contributions: $z(R_{e}) = \sum_{j=1}^{n} x_{j} v_{j}(r_{j})$. 
Consequently, the sketch vector $Y(R_{e}) = G^{\prime} \cdot z(R_{e})$ has a characteristic function that completely factorizes over the coordinates. 
Specifically, for any $\alpha \in \mathbb{R}^{R}$, we have
\[\Phi_{Y}(\alpha) = \mathbb{E}_{R_{e}}[\exp(2\pi i \langle \alpha, Y(R_{e}) \rangle)] = \prod_{j=1}^{n} \mathbb{E}_{r_{j}}[\exp(2\pi i \langle (G^{\prime})^{\top} \alpha, x_{j} v_{j}(r_{j}) \rangle)].\]
This matches the exact definition of an $(M, n)$-Fourier shape, where $F_{\alpha}(R_{e}) = \prod_{j=1}^{n} f_{j}(r_{j})$ and each $f_{j}$ maps into the complex unit disk. 
We apply the GKM PRG to fool this Fourier shape. 
For each coordinate $j$, simulating its updates requires drawing $\mathcal{O}(\log n)$ discrete exponential random variables and two discrete uniform random variables, which takes $\mathcal{O}(\log^2 n)$ bits, giving a domain size $M = 2^{\mathcal{O}(\log^2 n)}$.

The vector $Y$ has $R = \mathcal{O}(\log^2 n)$ coordinates, each discretized to $L = \O{\log n}$ bits, giving a support size $N_{supp} \le n^{\O{\log^2 n}}$. 
By the Cauchy-Schwarz inequality and Plancherel's theorem, the total variation distance between two probability distributions over a finite support is bounded by $\frac{1}{2}\sqrt{N_{supp}}$ times their maximum difference in the Fourier domain. 
Thus to guarantee the total variation distance changes by at most $\frac{1}{\poly(n)}$, it suffices to have Fourier error $\eps_{FT} \le \frac{n^{-\O{1}}}{\sqrt{N_{supp}}} = n^{-\O{\log^2 n}}$. 
By \Cref{thm:gkm:prg}, with $N=n$, $M=2^{\mathcal{O}(\log^2 n)}$, and $\eps=\eps_{FT}$, the corresponding seed length is $r_{2} = \O{\log \frac{NM}{\eps} (\log\log \frac{NM}{\eps})^2}$.

Since $\log(NM/\eps_{FT}) = \O{\log n} + \O{\log^2 n} + \O{\log^3 n} = \O{\log^3 n}$, the required seed length evaluates to exactly $r_{2} = \O{\log^3 n (\log \log n)^2}$. 
This pseudorandom generator $F_{2}$ ensures $d_{TV}(P_{Y}, P'_{Y}) \le \frac{1}{\poly(n)}$. 
Consequently,
\[|\mathbb{E}_{y_{1}, R_{e}}[\mathcal{A}_{i}(F_{1}(y_{1}), R_{e})] - \mathbb{E}_{y_{1}, y_{2}}[\mathcal{A}_{i}(F_{1}(y_{1}), F_{2}(y_{2}))]| \le \frac{1}{\poly(n)},\]
\end{proof}

\subsection{End-to-End Correctness}
\label{sec:final}
Finally, we have the end-to-end correctness of our algorithm, combined with the statistical test and the derandomization. 
\begin{restatable}[End-to-End Correctness for Perfect $L_2$ Sampling]{theorem}{thmendtoend}
\label{thm:end:to:end}
Given a dynamic turnstile stream defining a frequency vector $x\in\mathbb{R}^n$, there exists an algorithm using expected $\polylog(n)$ update time and $\tO{\log^3 n}$ bits of space such that, conditioned on returning an index $(i^*\neq\perp)$,
\[\PPr{i^*=j\mid i^*\neq\perp}=\frac{x_j^2}{\|x\|_2^2}\pm\O{\nu_{err}},\]
where $\nu_{err}=n^{-c}$ for any fixed constant $c>0$.
\end{restatable}
\begin{proof}
Let $\calD_{ideal}$ denote the distribution induced by the exact discrete construction, and let $\calD_{alg}$ denote the distribution produced by the implemented algorithm. 
Let $\calE_{pass}$ denote the event that the statistical test succeeds and the algorithm returns an index.

Under $\calD_{ideal}$, the dense $\CountSketch$ recovers the unique maximum coordinate of the duplicated vector $z\in\mathbb{R}^N$ except with probability $n^{-\O{c}}$. 
Since the gap test only accepts when the top coordinate is separated from the remainder of the vector, conditioning on $\calE_{pass}$ implies that the returned coordinate agrees with the true maximizer $D(1)$ with high probability.
Fix a virtual coordinate $v\in[N]$.  
By \Cref{lem:swz_4_2},
\[\PPr{\calE_{pass}\mid D(1)=v}=\PPr{\calE_{pass}}\pm n^{-\O{c}}.\]
By \Cref{lem:swz_4_7}, we have that $\PPr{\calE_{pass}}\ge 0.75$. 
Therefore,
\[\PPr{\calE_{pass}\cap\{D(1)=v\}}=\PPPr{\calD_{ideal}}{D(1)=v}\cdot\left(\PPr{\calE_{pass}}\pm n^{-\O{c}}\right).\]
Let $C_j$ denote the collection of virtual copies associated with coordinate $j\in[n]$. 
Summing over all such copies gives
\begin{align*}
\PPr{i^*=j}&=\sum_{v\in C_j}\PPr{\calE_{pass}\cap\{D(1)=v\}}\pm n^{-\O{c}} \\
&=\sum_{v\in C_j}\PPPr{\calD_{ideal}}{D(1)=v}\left(\PPr{\calE_{pass}}\pm n^{-\O{c}}\right)\pm n^{-\O{c}}.
\end{align*}
By \Cref{lem:ideal_dist},
\[\sum_{v\in C_j}\PPPr{\calD_{ideal}}{D(1)=v}=\frac{x_j^2}{\|x\|_2^2}.\]
Substituting this identity gives
\[\PPr{i^*=j}=\frac{x_j^2}{\|x\|_2^2}\left(\PPr{\calE_{pass}}\pm n^{-\O{c}}\right)\pm n^{-\O{c}}.\]
Conditioning on $\calE_{pass}$ and using $\PPr{\calE_{pass}}\ge 0.75$,
\[\PPPr{\calD_{ideal}}{i^*=j\mid\calE_{pass}}=\frac{x_j^2}{\|x\|_2^2} \pm\O{n^{-c}}.\]

We now compare $\calD_{ideal}$ and $\calD_{alg}$. 
By \Cref{lem:tvd:error}, the total variation distance between the ideal distribution $\mathcal{D}_{ideal}$ and the algorithmic distribution $\mathcal{D}_{alg}$ is at most $\O{n^{-c}}$. 
Finally, the GKM18 pseudorandom generator $\nu_{err}$-fools both the dense $\CountSketch$ queries and the Fourier evaluations used by the oracle while using only $\O{\log^3 n(\log\log n)^2}$ random bits, contributing at most $\nu_{err}/3$ additional variation distance.

Combining the three bounds gives
\[d_{TV}(\calD_{ideal},\calD_{alg})\le\nu_{err}.\]
It follows that
\[\PPPr{\calD_{alg}}{i^*=j\mid i^*\neq\perp}=\frac{x_j^2}{\|x\|_2^2}\pm\O{\nu_{err}}.\]

Finally, the dense $\CountSketch$ uses $M=\Theta(\log n)$ blocks with $r=\Theta(\log n)$ rows each, for a total of $\Theta(\log^2 n)$ sketch rows. 
Each row stores a single counter using $\O{\log n}$ bits, giving overall space complexity $\O{\log^3 n}$ bits. 
Each update involves $\O{\log^2 n}$ rows and performs only polylogarithmic additional work for the explicitly tracked head coordinates and simulated tail contribution, giving expected update time $\polylog(n)$.
\end{proof}

\section*{Acknowledgments}
Google DeepThink contributed the main ideas underlying the technical proofs in this work, and several other AI models assisted with verification. 

\bibliographystyle{alpha}
\bibliography{references}
\end{document}